\documentclass{vldb}
\usepackage{graphicx}
\usepackage{balance}  
\usepackage{caption}
\usepackage{subfig}
\usepackage{amsmath}
\usepackage{amssymb}
\usepackage{algorithm}
\usepackage{algorithmic}
\usepackage{color}
\usepackage{comment}
\usepackage{flushend}
\usepackage{bm}
\usepackage{soul}
\usepackage{multirow}
\usepackage{array}

\begin{document}


\title{Clustering Stream Data by Exploring the Evolution of Density Mountain}
\vspace{-0.1in}



\numberofauthors{1} 

\author{
%
%
\alignauthor
Shufeng Gong, Yanfeng Zhang, Ge Yu\\
      \affaddr{Northeast University, Shenyang, China}\\
       \email{gongsf@stumail.neu.edu.cn, \{zhangyf, yuge\}@mail.neu.edu.cn}
}
\date{30 July 1999}
\vspace{-0.1in}
\maketitle

\hyphenpenalty=5000

\begin{abstract}
Stream clustering is a fundamental problem in many streaming data analysis applications. Comparing to classical batch-mode clustering, there are two key challenges in stream clustering: (i) Given that input data are changing continuously, how to incrementally update clustering results efficiently? (ii) Given that clusters continuously evolve with the evolution of data, how to capture the cluster evolution activities? Unfortunately, most of existing stream clustering algorithms can neither update the cluster result in real time nor track the evolution of clusters.

In this paper, we propose an stream clustering algorithm \textit{EDMStream} by exploring the \textbf{E}volution of \textbf{D}ensity \textbf{M}ountain. The \emph{density mountain} is used to abstract the data distribution, the changes of which indicate data distribution evolution. We track the evolution of clusters by monitoring the changes of density mountains. We further provide efficient data structures and filtering schemes to ensure the update of density mountains in real time, which makes online clustering possible. The experimental results on synthetic and real datasets show that, comparing to the state-of-the-art stream clustering algorithms, e.g., D-Stream, DenStream, DBSTREAM and MR-Stream, our algorithm can response to a cluster update much faster (say 7-15x faster than the best of the competitors) and at the same time achieve comparable cluster quality. Furthermore, \textit{EDMStream} can successfully capture the cluster evolution activities.
\end{abstract}

\vspace{-0.1in}
\newtheorem{theorem}{theorem}
\newtheorem{lemma}{lemma}
\newtheorem{proposition}{proposition}

\newcommand{\oldcolor}[1]{\textcolor{blue}{{}#1}}
\newcommand{\revisioncolor}[1]{\textcolor{red}{{}#1}}

\section{Introduction}\label{sec:intro}
Recent advances in both hardware and software have resulted in a large amount of data, such as sensor data, stock transition data, news, tweets and network flow data etc. A kind of such data that continuously and rapidly grow over time are referred to as data streams \cite{aggarwal2007data}. Clustering stream data is one of the most fundamental problems in many streaming data analysis applications. Basically, it groups streaming data on the basis of their similarity, where data evolves over time and arrives in an unbounded stream.

Discovering of the patterns hidden in streams is substantial and essential for understanding and further utilizing these data, and there are large number of efforts contribute it, such as \cite{begum2014rare, cao2014interactive, huang2015streaming}. Take the news recommendation system as an example. The news recommendation system aims to present the news that will interest users. The news are clustered according to their similarities, so that the news in the same cluster as that a user has visited is recommended to the user. As the news are generated continuously, the news can be treated as a news stream. Furthermore, the news clusters are evolving as fresh news coming out and outdated news fading out. In order to make a timely recommendation, it is crucial to capture the cluster evolution and update the news clusters in realtime. In fact, stream clustering is widely used in a broad range of applications, including network intrusion detection, weather monitoring, and web site analysis. Due to its great importance, stream clustering has attracted many research efforts \cite{aggarwal2003framework, Cao2006Density, Chen2007Density, Isaksson2012SOStream, Silva2014Data, Zhang2014Data}.

Comparing to classical batch-mode clustering methods \cite{rodriguez2014clustering, Gong2016EDDPC, Zhang2016Efficient}, there are two additional key challenges in stream clustering. First, stream data are supposed to arrive in a high speed. In order to reflect changes of the underlying stream data, stream clustering algorithms are required to \emph{update clustering results quickly and frequently}. Second, multiple clusters might merge into a large cluster, and a single cluster might be split into multiple small clusters over time. In order to capture the cluster evolution activities, stream clustering algorithms are required to \emph{have the ability of tracking cluster evolution}.

For the first challenge, most existing solutions \cite{Cao2006Density,Chen2007Density} summarize data points in stream using summary structures (e.g., micro-clusters \cite{aggarwal2003framework,Cao2006Density}, grids\cite{Chen2007Density}) and update these summarizations upon receiving new points. Using the summary structures can reduce processing overhead. Then an offline batch-mode clustering algorithm is periodically performed on these summaries to update clustering result. However, these stream clustering algorithms are not designed for incremental update and are still very expensive to offline update clustering results. For the second challenge, they leverage an additional offline cluster evolution detecting procedure (e.g., MONIC \cite{Spiliopoulou2006MONIC} and MEC \cite{Oliveira2010MEC}). Due to the expensive offline clustering and offline tracking step, the existing solutions can neither update the clustering result in real time nor monitor the evolution of clusters in real time. Thus, a stream clustering algorithm that can update clustering result and monitor the cluster evolution in real time is desired.


In this paper, we propose a density-based stream clustering algorithm \textit{EDMStream}. We rely on the first assumption that cluster centers are surrounded by neighbors with lower local density\footnote{Points with lower local density means that these points are in a low density region.}. Then we can draw the density distribution of points as shown in Fig. \ref{mersplit}(a), which is referred to as \emph{density mountain}. The cluster center is at the mountain peak and the borderline points are at the mountain foot. Note that this is an illustrative figure and the points are in a \textit{1-dimension} space. In general, the density mountain should be drawn in a multi-dimensional plot. We rely on the second assumption that the center point has a relatively large distance from any other higher density points. As shown in Fig. \ref{mersplit}(b), there are two clusters corresponding to two density mountains, and there is a valley between two mountains. The right density mountain's peak has a relatively large distance to the higher density points, since the higher density points are located on left (higher) density mountain, while other points on the way up to the density peak are with relatively small distance to the higher density points. As a result, a wide density valley appears between two density mountains, and the nearest distance to higher density point (labeled as $\delta$ in Fig. \ref{mersplit}(b)) plays a key role in identifying clusters. The cluster evolution of data stream can be detected as long as the distance to the nearest higher density point is large enough or small enough.

To quickly update clustering result, we first summarize a set of close points as a \emph{cluster-cell} to reduce computation and maintenance cost. We then propose an efficient hierarchical tree-like structure \emph{Dependency-Tree (DP-Tree)} to abstract the density mountains. The DP-Tree maintains the relationships between cluster-cells (i.e., a cluster-cell and its nearest higher density cluster-cell) and implies the relationship between clusters. Our algorithm can quickly return the update clustering result by efficiently updating the DP-Tree structure. Meanwhile, by tracking the update of the DP-Tree structure, we can also track the cluster evolution activities.

\begin{figure}
  \centering
    \includegraphics[width=3in]{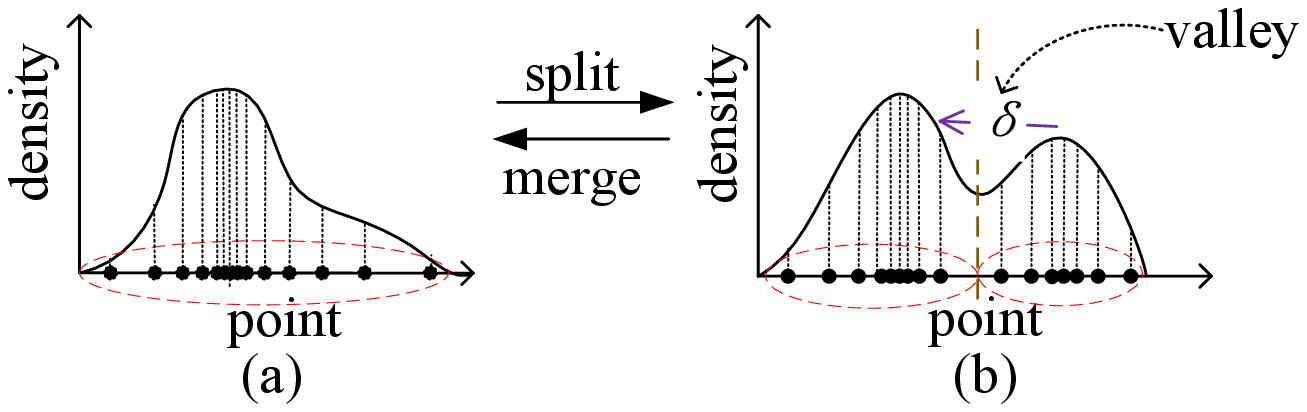}\\
    \vspace{-0.1in}
  \caption{The shape of density mountain changes as the (1-dimension) data distribution evolves.}\label{mersplit}
  \vspace{-0.1in}
\end{figure}


Another feature that distinguishes our algorithm from other existing solutions is the ability of adjusting itself and adapting to data distribution changes. Through a user-interaction step in the initialization phase, it can learn user preference to cluster granularity. In terms of the user preference, it can dynamically adjust the cluster separation strategy for the evolving data stream. This can greatly improve the quality of clustering result as shown in our experimental results.

We summarize the contributions of this paper as follows.

\begin{list}{\labelitemi}{\setlength{\leftmargin}{5mm}\setlength{\itemindent}{0mm}
\setlength{\topsep}{0.5mm}\setlength{\itemsep}{2mm}\setlength{\parsep}{0mm}}


\item \textbf{(Effectiveness on Cluster Evolution Tracking)} We propose a novel cluster evolution detection strategy by monitoring density mountains. Comparing to existing solutions, it can not only achieve comparable cluster quality but also tell how the clusters evolve.
\vspace{-0.02in}
\item \textbf{(Efficiency)} We design a highly efficient data structure Dependency-Tree to maintain the states of density mountains. By using two filtering strategies, a large amount of unnecessary tree update operations are avoided. It improves the performance a lot for cluster result updates and helps monitor cluster evolution efficiently.
\vspace{-0.02in}
\item \textbf{(Adaptability)} We provide an automatic adjusting strategy. It learns the user preference through an initial user-interaction step and automatically updates the algorithm parameters according to data evolution, which makes our clustering algorithm self-adaptive to data distribution changes.
\vspace{-0.02in}
\item We perform extensive experiments to evaluate \textit{EDMStream}. We compare \emph{EDMStream} with the state-of-art stream clustering algorithms, including D-Stream \cite{Chen2007Density}, DenStream \cite{Cao2006Density}, DBSTREAM \cite{hahsler2016clustering} and MR-Stream \cite{Wan2009Density}. Our results show that \emph{EDMStream} outperforms these algorithms on both effectiveness and efficiency. \emph{EDMStream} only takes 7-23$\mu$s for each cluster update. It exhibits 7-15x speedup over the other algorithms. We also demonstrate its adaptability to automatically adjust key parameters according to data evolution. Additionally, we also introduce a news recommendation use case and show its ability to track cluster evolution.
\end{list}

The remainder of this paper is organized as follows. Sec. \ref{sec:DP2} reviews the Density Peaks Clustering and proposes our Dependency-Tree structure. Sec. \ref{sec:streamcluster} presents our definitions and related concepts of stream clustering. In Sec. \ref{sec:edmstream}, we introduce \textit{EDMStream} algorithm. Sec. \ref{sec:adaptive} discusses \textit{EDMStream}'s self-adaptive strategy. Experimental results are shown in Sec. \ref{sec:expr}. Sec. \ref{sec:relate} reviews the related work and Sec. \ref{sec:conclu} concludes this paper.

\newcommand{\Paragraph}[1]{\smallskip\noindent{\bf #1.}}
\newcommand{\change}[1]{\textcolor{OliveGreen}{#1}}

\section{DP Clustering and DP-Tree}
\label{sec:DP2}
\vspace{-0.07in}
\subsection{DP Clustering}
\label{sec:DP}
Density Peaks (\textit{\textbf{DP}}) Clustering \cite{rodriguez2014clustering} is a novel clustering algorithm recently proposed by Rodriguez and Laio \cite{rodriguez2014clustering}. The algorithm is based on two observations: (i) cluster centers are often surrounded by neighbors with lower local densities, and (ii) they are at a relatively large distance from any points with higher local densities. Correspondingly, DP computes two metrics for every data point: (i) its \emph{local density} $\rho$ and (ii) its distance $\delta$ from other points with higher density. DP uses the two metrics to locate density peaks, which are the cluster centers.

\begin{figure}
\center
    \subfloat[Plane view]{\fbox{\includegraphics[width = 1.3in]{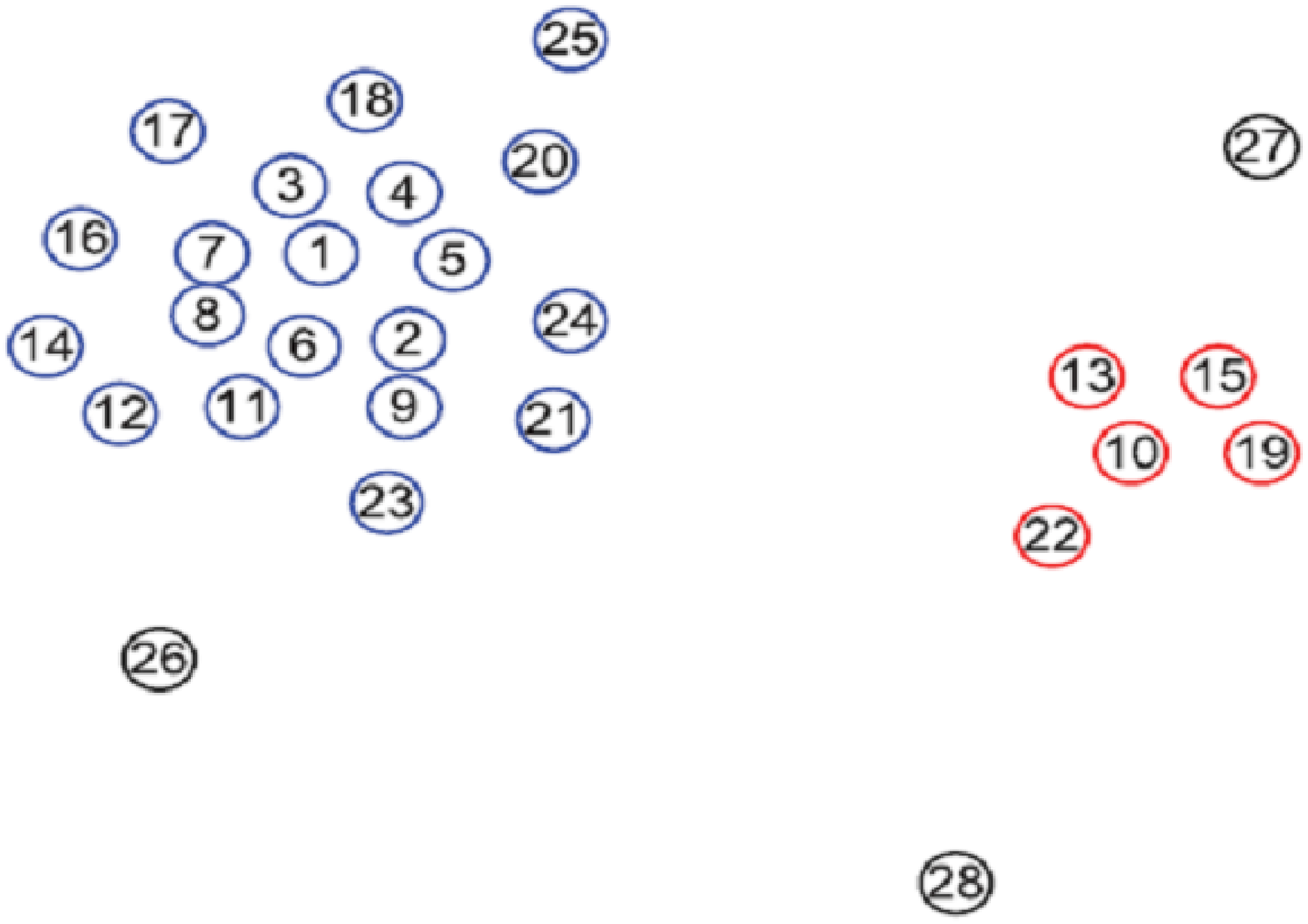}\label{dpa}}}
    \hspace{0.05in}
    \subfloat[Decision graph]{\fbox{\includegraphics[width = 1.3in]{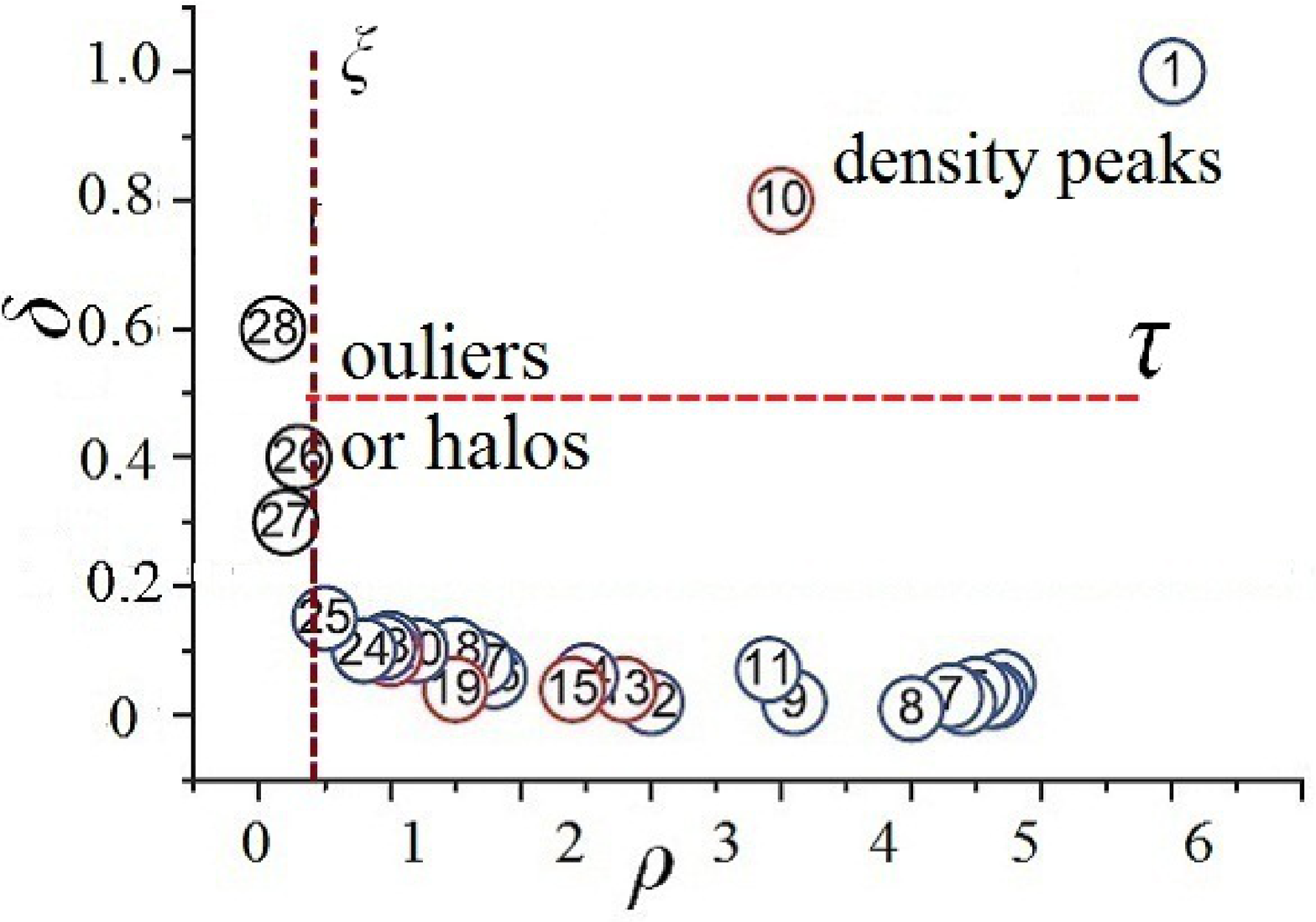}\label{dpb}}}
    \vspace{-0.1in}
\caption{Density Peaks Clustering  \protect\cite{rodriguez2014clustering}}\label{fig:DPClustering}
\vspace{-0.2in}
\end{figure}

The \emph{\textbf{local density}} $\rho_i$ of data point $p_i$ is the number of points whose distance to $p_i$ is smaller than $d_c$.
\vspace{-0.07in}
\begin{equation}\label{eq:rho}
\vspace{-0.07in}
  \rho_i=|\{p_j\big| |p_i,p_j| <d_c\}|
\end{equation}
\noindent where $|p_i,p_j|$ is the distance\footnote{In this paper, the distance means the Euclidean distance unless particularly mentioned.} from point $p_i$ to point $p_j$, and $d_c$ is called the cutoff distance. We use the density value to distinguish \emph{\textbf{outliers} whose density is no bigger than a predefined value $\xi$ (i.e., $\rho_i\leq\xi$).}

The \emph{\textbf{dependent distance}} $\delta_i$ of point $p_i$ is computed as
\begin{equation}\label{eq:delta}
\vspace{-0.05in}
  \delta_i=\min_{j:\rho_j>\rho_i}(|p_i,p_j|)
  \vspace{-0.05in}
\end{equation}
\noindent It is the minimum distance from point $p_i$ to any other point whose local density is higher than that of point $p_i$. Suppose point $p_j$ is point $p_i$'s the nearest neighbor\footnote{If multiple higher density nearest neighbors' distances are equal, we randomly pick one among them.} with higher density, i.e., $p_j=argmin_{j:\rho_j>\rho_i}(|p_i,p_j|)$. We say that point $p_i$ is \emph{\textbf{dependent}} on point $p_j$.

Let us think more about the local density $\rho$ and dependent distance $\delta$. A point with small $\rho_i$ is likely to be outliers no matter how large its $\delta_i$ is. Next, we focus on the points with relatively large $\rho_i$ to study the effect of $\delta_i$. Small $\delta_i$ implies that point $p_i$ is surrounded by at least one higher density neighbor. Anomalously large $\delta_i$ implies that point $p_i$ is far from another dense area and point $p_i$ itself could be the density peak of its own region, since it has no higher density neighbor. $\delta_i$ is much larger than the typical nearest neighbor distance only for points that are local or global maxima in the density. Thus, cluster centers are recognized as points for which the value of $\delta_i$ is anomalously large as well as large $\rho_i$. This is also illustrated in Fig. \ref{mersplit}, where the points in the same density mountain have relatively small dependent distance except for the density peak (i.e., cluster center).

If the dependent distance from a point $p_i$ to its dependency $p_j$ is no bigger than $\tau$ (i.e., $\delta_i\leq\tau$), we say it is \emph{\textbf{strongly dependent}}, otherwise it is \emph{\textbf{weakly dependent}}. For a set of points $\{p_1,p_2,\ldots,p_n\}$, there exist a strongly dependent chain such that point $p_i$ ($1\leq i\leq n-1$) is strongly dependent on $p_{i+1}$, where the end point $p_n$ is not strongly dependent on any other point (might be weakly dependent on other point). We call point $p_n$ as any $p_i$'s \emph{\textbf{strongly dependent root}}. Point $p_i$ ($1\leq i\leq j-1$) is \emph{\textbf{dependency-reachable to}} any $p_j$ ($i+1\leq j\leq n$). Then, a cluster in DP algorithm can be defined as follows:

\newdef{definition}{\textbf{Definition}}
\vspace{-0.05in}
\begin{definition}\label{def:cluster}
\emph{(\textbf{Cluster}) Let $P$ be a set of points. A cluster $C$ is a non-empty subset of $P$ such that:}
\begin{itemize}
\vspace{-0.05in}
  \item \emph{(Maximality) If a point $p\in C$, then any non-outlier point $q$ that is dependency-reachable to $p$ also belongs to $C$.}
      \vspace{-0.05in}
  \item \emph{(Traceability) For any points $p_1,p_2,\ldots\in C$, they have the same strongly dependent root, which is the density peak in $C$.}
      \vspace{-0.05in}
\end{itemize}
\end{definition}

Fig. \ref{fig:DPClustering} illustrates the process of Density Peaks Clustering through a concrete example. Fig. \ref{dpa} shows the distribution of a set of 2-D data points. Each point $p_i$ is depicted on a \emph{decision graph} by using ($\rho_i$, $\delta_i$) as its x-y coordinate as shown in Fig. \ref{dpb}. By observing the decision graph, the \emph{density peaks} can be identified in the top right region since they are with relatively large $\rho_i$ and large $\delta_i$ (i.e., $\rho_i>\xi$ and $\delta_i>\tau$). The \emph{outliers} or \emph{halos}\footnote{The cluster halos are the points that locate at the borders of clusters.} can be identified in the left region whose $\rho_i\leq\xi$. Since each point is only dependent on a single point, for each point there is a dependent chain ending at a density peak. After the density peaks (as cluster representatives) have been found, each remaining point is assigned to the same cluster as its dependent point.



\subsection{Dependency Tree (DP-Tree)}
\label{sec:DP-Tree}

To abstract the DP clustering, we propose a tree-like structure,  \emph{Dependency Tree}, which can track the correlations between points and between clusters. As mentioned in Sec. \ref{sec:DP}, the clustering process is achieved by tracking the dependency chain. The point-point dependency relationship implies the point-cluster correlations. In order to support online stream processing, an efficient data structure is desired to maintain the dependency relationship. Since each point is only dependent on a single point (except for the absolute density peak with the highest density), the point-point dependencies can be abstracted by a tree-like structure, which is denoted as \emph{Dependency Tree} (\textbf{\emph{DP-Tree}}). Fig. \ref{fig:dptree-point} shows an illustrative DP-Tree for the points shown in Fig. \ref{dpa}.


\begin{figure}
  \centering
    \includegraphics[width=3in]{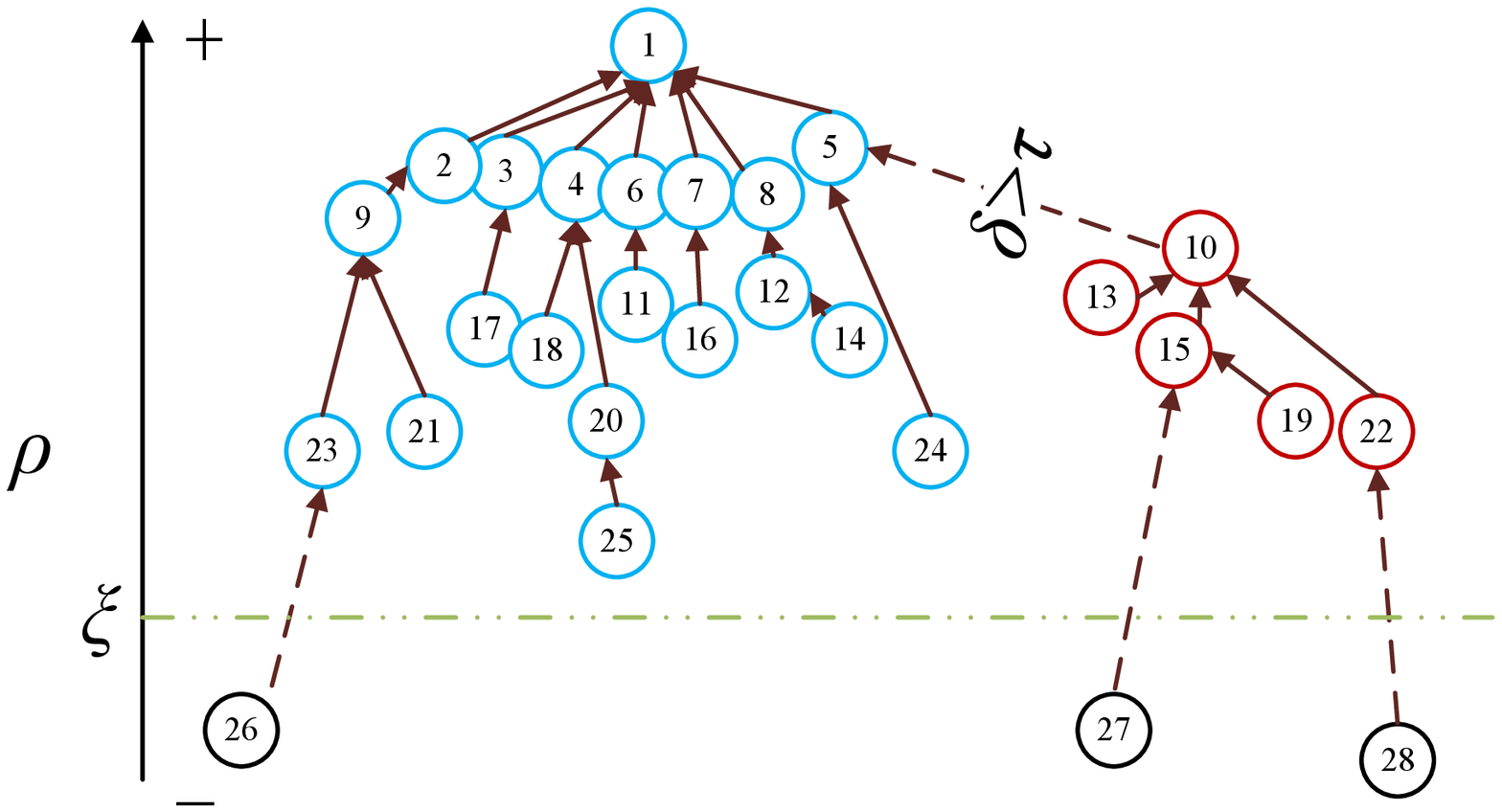}\\
    \vspace{-0.08in}
  \caption{An illustrative example of DP-Tree. The dependency relationship is captured by setting arrows between points. The length of arrows indicates the dependent distance ($\delta$). The solid arrows indicate strongly dependent relationship. The dashed arrows indicate weakly dependent relationship.
  The points residing at higher levels of DP-Tree are with higher densities ($\rho$). The root of DP-Tree is the absolute density peak with the highest local density.}
  \label{fig:dptree-point}
  \vspace{-0.1in}
\end{figure}

Let us first divide the DP-Tree into two parts, i.e., the upper part (in which each node's density is larger than $\xi$) and the lower part (in which each node's density is smaller than or equal to $\xi$). The nodes that belong to the lower part are simply recognized as outliers. In the upper part of DP-Tree, for a given subtree $T_i$ in DP-Tree, if all the links in the subtree are strongly dependent, $T_i$ is a \emph{\textbf{strongly dependent subtree}}. If there is no other strongly dependent subtree $T_j$ such that $T_i$ is a subtree of $T_j$, we say $T_i$ is a \textbf{\emph{Maximal Strongly Dependent SubTree (MSDSubTree)}}. Given the definition of MSDSubTree and the definition of cluster in Def. \ref{def:cluster}, the clustering based on DP-Tree is defined as follows

\begin{definition}\label{def:clusterdptree}
\vspace{-0.05in}
\emph{(\textbf{Clustering based on DP-Tree}) The clustering based on DP-Tree is to find all the MSDSubTrees. Every MSDSubTree corresponds to a cluster. The root of a MSDSubTree is the cluster center of that cluster. }

\vspace{-0.05in}
\end{definition}

The DP-Tree structure is highly efficient for maintaining volatile clusters. It can quickly response to a cluster update query and can be used for tracking cluster evolutions due to its \emph{hierarchical} structure. Once a new point arrives, it is directly linked to its dependent point in terms of its local density. More importantly, the new point may affect its nearby points and change their cluster assignments. Many affected points can share the same predecessor, and they belong to the same cluster as their predecessor. We only need to change a pointer (the predecessor's pointer to another MSDSubTree) to complete this update, which greatly saves the maintaining cost.

\vspace{-0.05in}
\subsection{DP Clustering vs. DBSCAN}
\label{sec:dp:vsdbscan}

Due to the fact that the idea of DP clustering is pretty similar to that of DBSCAN, it is necessary to highlight the difference between them. The cluster defined in DBSCAN satisfies two criteria: \emph{maximality} and \emph{connectivity} \cite{ester1996density, gan2015dbscan}, while the cluster defined in DP clustering satisfies \emph{maximality} and \emph{traceability} (Def. \ref{def:cluster}). The maximality defined in both the two algorithms depicts the ``reachable'' property between two points, where the reachable property in both algorithms relies on point density information. However, the connectivity in DBSCAN depicts the density-connected property which is \emph{symmetric}, while the traceability in DP depicts the density-dependent property which is \emph{non-symmetric}. Therefore, in DBSCAN, the density-connected relationship between points can be abstracted as an undirected graph, and each connected component of the graph constitutes a preliminary cluster \cite{gan2017dynamic}. The clustering in DBSCAN is to find all the connected components from the density-connected graph. While in DP clustering, the density-dependent relationship between points can be abstracted with a tree-like structure (DP-Tree), and each MSDSubTree in the DP-Tree constitutes a cluster. The clustering in DP is to find all the MSDSubTrees from the DP-Tree.


\begin{figure}[tp]

  \centering
  \includegraphics[width=2.5in]{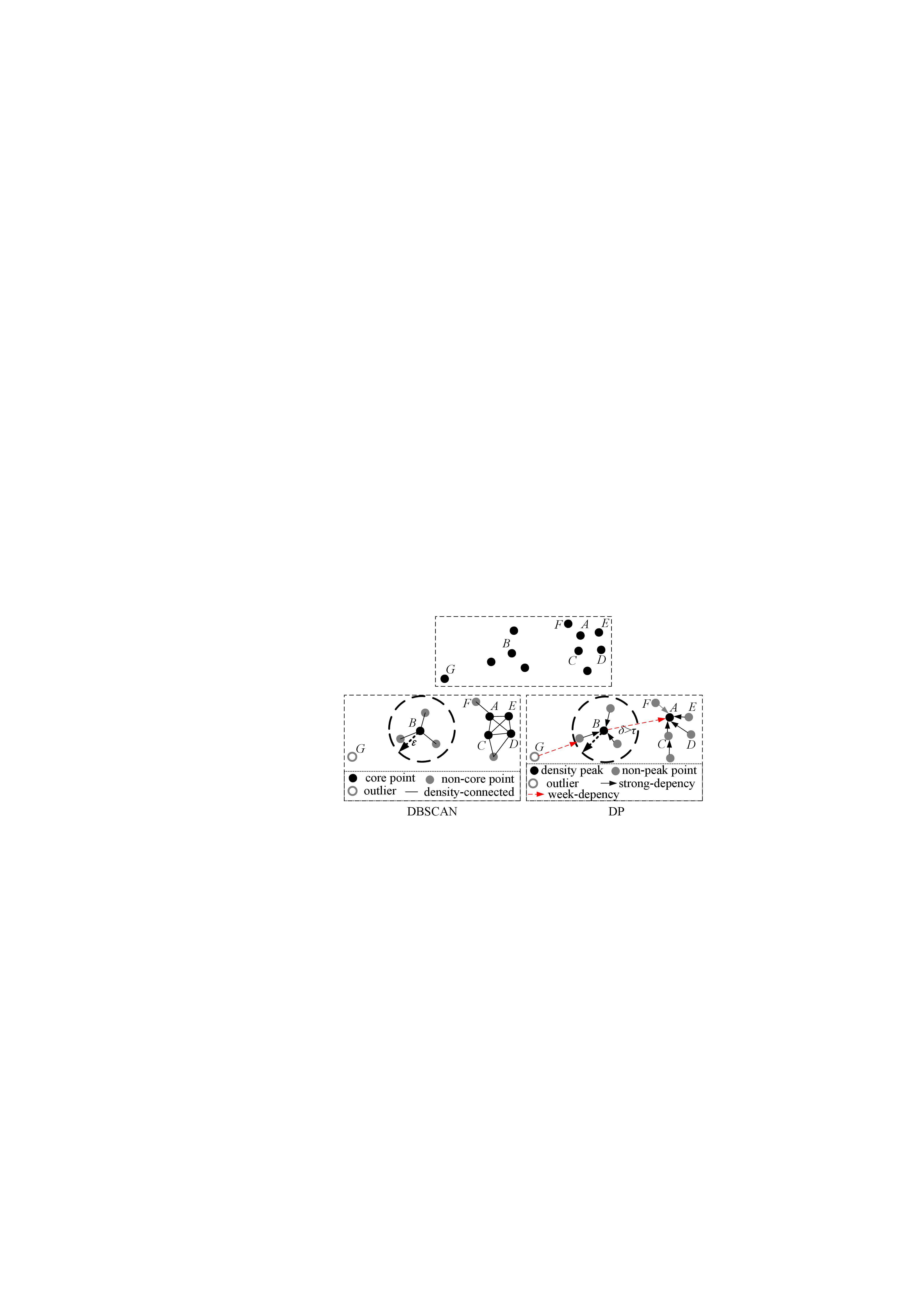}\\
  \vspace{-0.08in}
  \caption{DBSCAN vs. DP}\label{fig:DBSCAN-DP}
  \vspace{-0.25in}
\end{figure}

Fig. \ref{fig:DBSCAN-DP} provides an example to demonstrate the difference between DBSCAN and DP. The data distribution is shown on the upper side of the figure. DBSCAN first identifies the high density points as core points (e.g., points \textit{A-E} are core points because their densities are higher than a threshold), and connects these core points if they are closer than a threshold $\varepsilon$ (e.g., points \textit{A, C, D, E} are connected with each other, and point \textit{B} is not connected to the other core points because \textit{B} is far away from them). For each non-core point, it is connected to only one core point if they are close enough (e.g., point \textit{F}), otherwise is treated as an outlier (e.g., point \textit{G}). DBSCAN constructs such an undirected graph and aims to find all the connected components, each connected component corresponding to a cluster.

DP creates dependency connection from each point to its nearest neighbor with higher density (e.g., \textit{B-F} all depends on \textit{A} since \textit{A} is their nearest higher density neighbor). Different from DBSCAN's density-connected undirected graph, in DP each point depends only on a single point, and the dependency connection is directed. Thus, DP constructs a dependency tree structure (DP-Tree) rather than an undirected graph. In addition, DP distinguishes the weakly dependency connections whose distances are longer than a threshold $\tau$ (e.g., the dependency connection from \textit{B} to \textit{A}), and uses these weakly dependent connections to separate the DP-Tree into multiple subtrees (e.g., the subtree rooted from \textit{B} is separated from the subtree rooted from \textit{A} because their connection is weakly dependent). Each MSDSubTree corresponds to a cluster, and the root of each subtree is recognized as the density peak (e.g., points \textit{A} and \textit{B}). The points with extremely low density are recognized as outliers (e.g., point \textit{G}).

\vspace{-0.08in}
\section{Problem Statement}
\label{sec:streamcluster}

We aim to discover the potential clusters existing in data stream based on the two observations. 1) Dense areas are separated from each other by sparse areas; 2) Recent arrival data points play more important role in cluster representation than outdated data points. In this section, in terms of timeliness and unlimitedness of streams, we introduce the basic conceptions that will be used in stream clustering.

\vspace{-0.05in}
\subsection{Basic Conceptions}
\label{sec:streamcluster:stream}

\Paragraph{Data Stream} A data stream $S$ is a sequence of data points with timestamp information $p_1^{t_1}$, $p_2^{t_2}$, ..., $p^{t_N}_N$, i.e., $S^N=\{{p^{t_i}_i}\}^N_{i=1}$, which is potentially unbounded ($N\rightarrow\infty$). Each data point is described by a $d$-dimensional attribute vector with its arrival timestamp $t_i$.


\Paragraph{Decay Model} In most cases, the recent information from a stream reflects the emerging of new trends, e.g., weather monitoring and stock trade. The importance(freshness) of data should be decayed over time, so that the evolving characteristics of the stream can be captured. A common solution is to weight data points with an exponential time-dependent decay function \cite{Chen2007Density, Cao2006Density, Kranen2011The}. The freshness of point $i$ at time $t$ is
\vspace{-0.1in}
\begin{equation}\label{eq:freshness}
    f^{t}_i = a^{\lambda({t}-{t_i})}.
\end{equation}

\noindent This is a widely used decay function in many stream clustering algorithms \cite{Cao2006Density, Isaksson2012SOStream, Chen2007Density}. The parameter \textit{a} and $\lambda$ control the form of decay function. The higher the absolute value of $\lambda$ is, the faster the algorithm ``forgets'' old data. In this paper, we choose $a=0.998$, $\lambda = 1$ such that $f^{t}_i$ is in the range $(0,1]$. Suppose $\{p_j|t_j<t, |p_i,p_j|<d_c\}$ is a set of existing points whose distances to $p_i$ are smaller than $d_c$ before time point $t$. Point $p_i$'s local density at time $t$ is the sum of nearby points' $f^{t}_j$ rather than the number of nearby points as depicted in Equation \protect(\ref{eq:rho}).
\vspace{-0.05in}
\begin{equation}\label{eq:streamdensity}
    \rho_i^t=\sum_{p_j:t_j<t, |p_i,p_j|<d_c}f^{t}_j
    \vspace{-0.07in}
\end{equation}


The decay model implies that 1) if no new nearby point arrives the point density is decreasing over time and 2) all stream points are decaying at the same pace. In other words, we have a decay function $\mathcal{D}^t()$ applied on the current stream $S^n$ at any time $t$ to decay the streamed points.
\begin{equation}\label{eq:decay}
    \mathcal{D}^t(S^n)=\{f^{t}_1,f^{t}_2,\ldots,f^{t}_{n}\}
    \vspace{-0.05in}
\end{equation}

\Paragraph{Stream Clustering} Under the decay model, stream clustering is defined as follows.

\begin{definition}
\emph{(\textbf{Stream Clustering})} Given a data stream $S^N$ and their decayed freshness $\mathcal{D}^t(S^N)$, stream clustering $\mathcal{C}^t()$ returns a set of disjoint clusters at any time $t_1$, $t_2$, $\ldots$, $t_N$. That is, for any $n$ ($1\leq n\leq N$),  we have $\mathcal{C}^{t_n}\big(S^n, \mathcal{D}^{t_n}(S^n)\big)$ = $\{C^{t_n}_1$, $C^{t_n}_2$,$\ldots$, $C^{t_n}_{k^{t_n}}$, $C^{t_n}_o\}$, where $C^{t_n}_i$ $(1\leq i\leq k^{t_n})$ is a subset of $S^n$ at time $t_n$, $C^{t_n}_o$ is the set of outliers at time $t_n$, $k^{t_n}$ is the number of clusters at time $t_n$, $S^n=C^{t_n}_1\cup C^{t_n}_2\cup\ldots\cup C^{t_n}_{k^{t_n}}\cup C^{t_n}_o$, $C^{t_n}_i\cap C^{t_n}_{j}=\emptyset$, and $C^{t_n}_i\cap C^{t_n}_o=\emptyset$ for any $i$ and $j$.
\end{definition}
\vspace{-0.05in}
\Paragraph{Cluster Evolution} The clusters evolve continuously, i.e., $\mathcal{C}^{t_n}\big(S^n, \mathcal{D}^{t_n}(S^n)\big)\neq \mathcal{C}^{t_{n+1}}\big(S^{n+1}, \mathcal{D}^{t_{n+1}}(S^{n+1})\big)$. Specifically, the number of clusters may change, and the point-to-cluster assignment may change. By referring to the previous work \cite{Spiliopoulou2006MONIC, Oliveira2010MEC}, we define five types of evolutions which are summarized in Table \ref{tab:defevolution}.

\begin{table}[htb]
\vspace{-0.1in}
    \caption{Cluster evolution types}
    \vspace{-0.1in}
    \label{tab:defevolution}
    \centering
    \footnotesize
    \begin{tabular}{|l|l|}
    \hline
    \textbf{Type} & \textbf{Mathematical Notation} \\
    \hline
    \textit{Emerge} & $\emptyset\rightarrow C^{t_{n+1}}_i$\\
    \hline
    \textit{Disappear} & $C^{t_n}_i \rightarrow \emptyset$ \\
    \hline
    \textit{Split} & $C^{t_n}_i \rightarrow \{ C^{t_{n+1}}_{i_1},\ldots, C^{t_{n+1}}_{i_x} \}$ \\
    \hline
    \textit{Merge} & $\{ C^{t_{n}}_{i_1},\ldots, C^{t_{n}}_{i_x} \} \rightarrow C^{t_{n+1}}_i$ \\
    \hline
     \multirow{7}{*}{\textit{Adjust}} & \quad$C^{t_n}_i\rightarrow C^{t_{n+1}}_i$, $C^{t_n}_j\rightarrow C^{t_{n+1}}_j$ \\
         & 1. $C^{t_{n+1}}_i=C^{t_n}_i\backslash\{p_1,\ldots,p_l\}$, \\
         & \quad $C^{t_{n+1}}_j=C^{t_n}_j\cup\{p_1,\ldots,p_l\}$ \\
         & 2. $C^{t_{n+1}}_o=C^{t_n}_o\backslash\{p_1,\ldots,p_l\}$, \\
         & \quad $C^{t_{n+1}}_i=C^{t_n}_i\cup\{p_1,\ldots,p_l\}$ \\
         & 3. $C^{t_{n+1}}_i=C^{t_n}_i\backslash\{p_1,\ldots,p_l\}$, \\
         & \quad $C^{t_{n+1}}_o=C^{t_n}_o\cup\{p_1,\ldots,p_l\}$ \\
    \hline
   \end{tabular}
   \vspace{-0.15in}
\end{table}

The \textbf{\emph{emerge}} evolution means a new cluster's birth. The \textbf{\emph{disappear}} evolution means an old cluster's death. The \textbf{\emph{split}} evolution means that a cluster is split into two or more clusters. The \textbf{\emph{merge}} evolution means that two or more clusters merge into one cluster. The \textbf{\emph{adjust}} evolution happens when 1) some points move from one cluster to another cluster; 2) some outliers become reachable and are merged to a cluster; 3) some marginal points in a cluster become outliers. The first four types will change the number of clusters, while the last one only changes the point-to-cluster assignments. Note that, the first four evolutions might occur along with cluster adjustment, and the three kinds of adjustment might occur concurrently.

\vspace{-0.05in}
\subsection{Stream Data Summarization}
\label{sec:streamcluster:sum}

If stream data are massive or even unlimited, it is not possible to store all data in main memory. Therefore, it is necessary to summarize stream data in an efficient way. We summarize a set of close points as a \textbf{\emph{cluster-cell}} so as to reduce the memory/computation cost. The cluster-cell is formally defined as follows.

\begin{definition}
\vspace{-0.1in}
\emph{(\textbf{cluster-cell}) A cluster-cell $c$ summarizes a group of close points, which can be described by a three-tuple $\{s_c, \rho^t_c, \delta^t_c\}$ at time $t$. }
\begin{itemize}
  \item \emph{$s_c$ is the \textbf{\emph{seed point}} of a cluster-cell $c$. The cluster-cell $c$ seeded by $s_c$ summarizes a set of points whose distance to $s_c$ is less than the distance to any other seed point and is less than or equal to a predefined radius $r$, i.e., $P_c=\{p_i:s_c=\underset{s_k\in S_{seed}}{\arg\min}(|p_i,s_k|),|p_i,s_c|\leq r\}$ where $S_{seed}$ is the set containing all seed points.}
\vspace{-0.1in}
  \item \emph{$\rho^t_c$ is the summarization of all cluster-cell points' \textbf{\emph{timely density}} (abbrv. density) at time $t$, which is defined as follows.}
\vspace{-0.1in}
\begin{equation}
\label{eq:timedensity}
    \rho^t_c=\sum_{p_i\in P_c}f{^t_{i}}.
    \vspace{-0.1in}
\end{equation}

\emph{where $f{^{t}_{i}}$ is the freshness of $p_{i}$ at time $t$ defined in Equation (\ref{eq:freshness}).}
  \item \emph{$\delta^t_c$ is the \textbf{\emph{dependent distance}} from $s_c$ to its nearest cluster-cell seed point with higher cluster-cell density. Similar to Equation (\ref{eq:delta}), $\delta^t_c$ is defined as follows.}
\begin{equation}\label{eq:celldelta}
\vspace{-0.1in}
    \delta^t_c=\min_{c':\rho^t_{c'}>\rho^t_c}(|s_c, s_{c'}|).
    \vspace{-0.05in}
\end{equation}
\end{itemize}

\end{definition}

We take cluster-cell as the basic processing unit instead of point. In other words, we will operate on a DP-Tree where each node is a cluster-cell instead of a data point. A new arrival point is not directly inserted to the DP-Tree but could cause generating a new cluster-cell or increasing an existing cluster-cell's density. Both can lead to DP-Tree's update. On the other hand, the decaying of points will lead to the decaying of cluster-cells, which can also lead to DP-Tree's update. By using cluster-cell, we can approximately obtain the timely density of local regions and significantly reduce the memory/computation cost.


\vspace{-0.05in}
\subsection{Basic Ideas}\label{subsec:idea}
\vspace{-0.05in}
\Paragraph{Stream Clustering using DP-Tree} In the context of DP-Tree, stream clustering is simply to find all MSDSubTrees from a \emph{\textbf{dynamic}} DP-Tree. The DP-Tree is dynamic since new arrival points and decay model may cause tree structure's update.

\Paragraph{Evolution Tracking using DP-Tree} In addition, cluster evolution can be tracked by monitoring how the DP-Tree changes. $\romannumeral1$)Cluster \textit{emergence/disappearance} can be tracked by finding new generated/disappeared MSDSubTrees. $\romannumeral2$)Cluster \textit{split} can be tracked when an MSDSubTree is split into multiple MSDSubTrees (one or more dependent links become longer than $\tau$). $\romannumeral3$)Cluster \textit{merging} can be tracked when multiple MSDSubTrees merge into one MSDSubTree (one or more dependent links become shorter than $\tau$). Cluster adjustment can be tracked by that 1) multiple cluster-cells from an MSDSubTree are relinked to other MSDSubTrees; 2) multiple cluster-cells' densities become larger than $\xi$ and they are included in the MSDSubTrees that have their dependencies; 3) multiple cluster-cells' densities becomes smaller than $\xi$ and they are removed from their original MSDSubTrees.

\Paragraph{Our Goal} To sum up, we aim to design an algorithm that can efficiently maintain and monitor the dynamic DP-Tree and return the MSDSubTrees quickly upon any change.



\vspace{-0.08in}
\section{{\secit EDMStream}}
\label{sec:edmstream}
In this section we propose \emph{EDMStream} for clustering streaming data.

\subsection{Algorithm Overview}\label{sec:overview}
\emph{EDMStream} distinguishes itself from other existing stream clustering algorithms on the ability of updating clusters in real time and tracking cluster evolution. We briefly overview the \textit{EDMStream} algorithm in the following.




\begin{figure}
  \centering
  \includegraphics[width=2.5in]{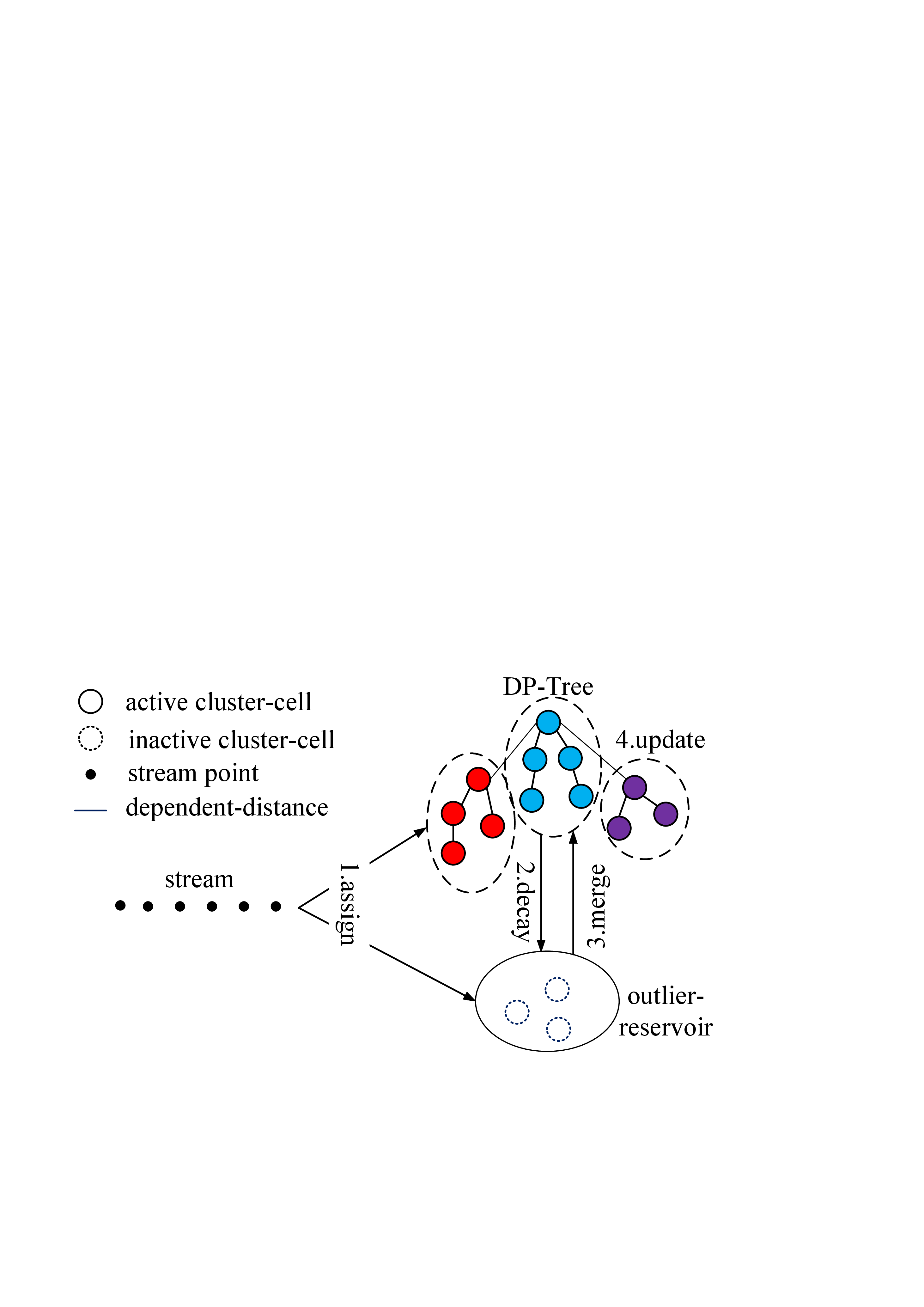}\\
  \vspace{-0.08in}
  \caption{\textit{EDMStream} Overview. }
  \label{process}
  \vspace{-0.2in}
\end{figure}

\Paragraph{Storage Structures} As shown in Fig. \ref{process}, two key storage structures are designed in \textit{EDMStream}.

1) DP-Tree. \emph{DP-Tree} is the data structure for abstracting density mountain. Each node in DP-Tree is a cluster-cell rather than a single point, which is for saving memory space and computation time as mentioned in Sec. \ref{sec:streamcluster:sum}.

2) Outlier Reservoir. Due to the unlimitedness of stream, it is desirable to limit the size of DP-Tree in order to reduce maintenance overhead. On the other hand, due to the evolution of stream, the role of clusters and outliers may change. The outliers may form a new cluster since they absorbs new arrival points. The old cluster may decay to outliers as they have not absorbed points for long time. Therefore, we use an \emph{outlier reservoir} for caching the cluster-cells with relatively low timely-density (i.e., temporal outliers), which are temporally not considered for clustering. Note that, a cluster-cell is moved to the outlier reservoir either because it contains only a few points (i.e., low local density) or because the contained points are outdated (low timely-density). The cluster-cells in the outlier reservoir are possible to absorb new points and be inserted to DP-Tree for clustering again.

\Paragraph{Key Operations} \textit{EDMStream} relies on four key operations as shown in Fig. \ref{process}.

1) New point assignment. A new point from stream is assigned to an existing cluster-cell (in DP-Tree or outlier reservoir) or forms a new outlier cluster-cell. A point $p_i$ is assigned to a cluster-cell $c$ if both the two conditions are satisfied: 1) the distance to the cluster seed $s_c$ is smaller than or equal to $r$, i.e., $|p_i,s_c|\leq r$; 2) $s_c$ is the closest cluster seed, i.e., $s_c=\underset{s_k\in S_{seed}}{\arg\min}(|p_i,s_k|)$
where $S_{seed}$ is the set of all existing cluster-cell seeds. If no such cluster-cell exists, a new cluster-cell seeded by $p_i$ is created and cached in outlier reservoir due to its low density.

2) Dependencies update. Due to the fading property, the densities of cluster-cells decrease over time. In addition, some of the cluster-cells might absorb new points. As a result, their densities increase, and their dependencies may change such that the DP-Tree is updated. We will present the details of dependencies update in Sec. \ref{sec:stream:update}.

3) Cluster-cell emergence (DP-Tree insertion). The timely-density of an existing cluster-cell in outlier reservoir increases after it absorbs a new point. It might be inserted into the DP-Tree for clustering once its density is large enough.

4) Cluster-cell decay (DP-Tree deletion). The timely-density of cluster-cells decays as the freshness of points is fading. If the density of a cluster-cell in DP-Tree is low enough, it might be temporally moved to the outlier reservoir. We will present the details of cluster-cell emergence/decay in Sec. \ref{sec:stream:emergedecay}.

\Paragraph{Cluster Evolution Tracking} The operations 2), 3), and 4) lead to cluster evolution. As described in Sec. \ref{subsec:idea}, we can track the evolution by monitoring the update of DP-Tree structure, including the changes of dependent distance (which will trigger the split or merge of MSDSubTrees), the insertion/deletion of cluster-cell nodes, and the movements of cluster-cell nodes between MSDSubTrees. The DP-Tree update operation along with the update time are then recorded for future queries.

\Paragraph{Initialization} Initially, a number of cluster-cells that absorb incoming points are cached in memory. Once the size of cached cluster-cells exceeds a predefined threshold, the density and the dependent distance of each cached cluster-cell are calculated in terms of Equation (\ref{eq:timedensity}) and Equation (\ref{eq:celldelta}) respectively. In the meantime, the dependencies of the cached cluster-cells are retrieved, which are used for initializing the DP-Tree structure. Furthermore, given $\tau$ a primary clustering result can be obtained as mentioned in Sec. \ref{sec:DP-Tree}. Next, we propose the techniques for efficiently updating the DP-Tree as new points coming.

\vspace{-0.05in}
\subsection{Dependencies Update (DP-Tree Update)}
\label{sec:stream:update}

The densities of all cluster-cells gradually decay as time goes by. But if cluster-cells absorb points, their densities should be increased. Moreover, the dependency relationship should also be updated accordingly.

\Paragraph{Densities Update} The decayed density of summary structures has been well studied in the literature \cite{Cao2006Density, Chen2007Density, ester1996density, Isaksson2012SOStream, Silva2014Data, Wan2009Density}. Based on our time decay model, if a cluster-cell absorbs a points from $t_{j}$ to $t_{j+1}$, its density is updated as follows. The proof can be referred to \cite{zhang1996birch}.
\begin{equation}\label{eq:update}
\vspace{-0.1in}
    \rho^{t_{j+1}}_c=a^{\lambda(t_{j+1}-t_j)}\rho^{t_j}_c+1.
\end{equation}

\Paragraph{Dependencies Update}
The changes of densities may cause the dependency changes. Suppose the density $\rho_c$ of cluster-cell $c$ increases and it becomes larger than the density of $c$'s original dependent cluster-cell. This means that $c$'s dependent cluster-cell and its dependent distance $\delta_c$ should be updated according to Equation (\ref{eq:celldelta}). It is also possible that $c$ becomes the new dependency of other some cluster-cells, whose dependent distances should also be updated accordingly. Every time $c$ absorbs new point, a large number of cluster-cells are involved in dependencies update. This can result in great computational burden, which poses challenge to real-time stream processing.

Let $\rho_c^{t_j}$ be the density of cluster-cell $c$ at time $t_j$. Let $F_c^{t_j}=\{c'|\rho_c^{t_j}<\rho_{c'}^{t_j}\}$ be the set of cluster-cells whose density are higher than $c$ at time $t_j$. In the DP-Tree point of view, $F_c^{t_j}$ is the set of cluster-cells that are at higher levels than $c$ at time $t_j$. We define $D_c^{t_j}$ as the \emph{\textbf{dependent cluster-cell}} of $c$ at time $t_j$.
\vspace{-0.1in}
\begin{equation}\label{eq:depcluster}
D_c^{t_j}=\arg\min_{c':c'\in F_c^{t_j}}|s_c,s_{c'}|.
\vspace{-0.1in}
\end{equation}

Since all cluster-cells' densities decay at the same rate, the order of their densities will not change from $t_j$ to $t_{j+1}$ except for the cluster-cell $c'$ that absorbs new point. For each $c$, we just need to judge whether the updated $c'$ newly appears in $F_c^{t_{j+1}}$ with respect to $F_c^{t_j}$. If so, $c$'s dependency update is required, otherwise can be avoided. This is because that, according to (\ref{eq:depcluster}), as long as set $F_c$ is consistent, $c$'s dependency $D_c$ will not change. From $c'$'s perspective, only the nodes whose density are previously higher than or equal to $\rho_{c'}$ ($\rho_c^{t_j}\geq\rho_{c'}^{t_j}$) but currently lower than $\rho_{c'}$ ($\rho_c^{t_{j+1}} < \rho_{c'}^{t_{j+1}}$) are necessary to update dependencies. In the DP-Tree point of view, we are trying to find out the nodes that are previously are at higher levels than $c'$ but now at lower levels, and only update their dependencies. Therefore, in order to reduce update cost, we propose our first density filtering scheme through the following theorem.

\vspace{-0.1in}
\begin{theorem}\label{theo:upnum}
\textbf{(Density Filter)} Suppose another cluster-cell $c'$ absorbs a point at time $t_{j+1}$.
\begin{equation}
\vspace{-0.05in}
\begin{aligned}
\text{If } \rho_c^{t_j}<\rho_{c'}^{t_j} &\text{ or } \rho_c^{t_{j+1}}\geq\rho_{c'}^{t_{j+1}}, \text{then}\\
     D_c^{t_j}&=D_c^{t_{j+1}}.\nonumber
\end{aligned}
\vspace{-0.05in}
\end{equation}
That is, cluster-cell $c$'s dependency will not change and it is not necessary to update $c$'s dependencies.
\end{theorem}
\vspace{-0.1in}

\begin{proof}


If $\rho_c^{t_j}<\rho_{c'}^{t_j}$, then $c'\in F_c^{t_j}$. After $c'$ absorbs new point at $t_{j+1}$ we still have $\rho_c^{t_{j+1}}<\rho_{c'}^{t_{j+1}}$ and $c'\in F_c^{t_{j+1}}$. $c'$ appears in both $F_c^{t_j}$ and $F_c^{t_{j+1}}$, i.e., $F_c^{t_j}=F_c^{t_{j+1}}$. Hence, $D_c^{t_j}=D_c^{t_{j+1}}$.

If $\rho_c^{t_{j+1}}\geq\rho_{c'}^{t_{j+1}}$, then $c'\notin F_c^{t_{j+1}}$. Even after $c'$ absorbs new point at time $t_{j+1}$ $\rho_c^{t_{j+1}}\geq\rho_{c'}^{t_{j+1}}$, so at time $t_j$ $\rho_c^{t_j}>\rho_{c'}^{t_j}$, i.e., $c'\notin F_c^{t_{j}}$. $c'$ neither appears in $F_c^{t_j}$ nor $F_c^{t_{j+1}}$, i.e., $F_c^{t_j}=F_c^{t_{j+1}}$. Hence, $D_c^{t_j}=D_c^{t_{j+1}}$.
\end{proof}

In addition, we exploit the triangle inequality property and propose our second filtering scheme through the following theorem.

\vspace{-0.1 in}
\begin{theorem}\label{theo:upcell}
\textbf{(Triangle Inequality Filter)} Suppose another cluster-cell $c'$ absorbs a point $p$ at time $t_{j+1}$.
\begin{equation}
\vspace{-0.05in}
\begin{aligned}
\text{If } \big||p,s_c|&-|p,s_{c'}|\big|>\delta^{t_j}_c, \text{then}\\
     D_c^{t_j}&=D_c^{t_{j+1}}.\nonumber
\end{aligned}
\vspace{-0.05in}
\end{equation}
That is, cluster-cell $c$'s dependency will not change and it is unnecessary to update $c$'s dependencies.
\end{theorem}
\vspace{-0.1in}

\begin{proof}
In terms of triangle inequality, $|s_c,s_{c'}|>\big||p,s_c|-|p,s_{c'}|\big|$. If $\big||p,s_c|-|p,s_{c'}|\big|>\delta^{t_j}_c$, then $|s_c,s_{c'}|>\delta^{t_j}_c$. According to the definition of dependent cluster-cell, $D_c^{t_j}$ is the \emph{nearest} higher density cluster-cell. Therefore, it is not possible to replace $c$'s original dependent cluster-cell by $c'$. Then we have $D_c^{t_j}=D_c^{t_{j+1}}$.
\end{proof}
\vspace{-0.05in}

According to Theorem \ref{theo:upnum}, we can avoid the dependency update of $c$ if $\rho_c^{t_j}<\rho_{c'}^{t_j}$ or $\rho_c^{t_{j+1}}\geq\rho_{c'}^{t_{j+1}}$. According to Theorem \ref{theo:upcell}, we can further reduce the number of dependency updates if $\big||p,s_c|-|p,s_{c'}|\big|>\delta^{t_j}_c$. Since $|p,s_c|$ and $|p,s_{c'}|$ have already been measured during the point assignment phase, the filtering cost is almost free. Our experimental results in Sec. \ref{sec:uptime} will show that a significant performance improvement is achieved.

\vspace{-0.05in}
\subsection{Cluster-Cells Emergence and Decay}
\label{sec:stream:emergedecay}

As discussed in Sec. \ref{sec:overview}, an outlier reservoir caching low timely-density cluster-cells is designed to maintain the outliers or halos. We only consider the dense cluster-cells for clustering. This is because that dense regions are more representative to reflect stream trends, and losing sight of sparse regions is helpful for distinguishing the true clusters. However, considering that the low density cluster-cells may become dense as they absorb new points, it is not a good idea to delete them immediately. Accordingly, we preserve these low density cluster-cells in outlier reservoir temporarily. We call the cluster-cells residing in DP-Tree as \emph{active} cluster-cells and the ones in outlier-reservoir as \emph{inactive} cluster-cells.

For ease of exposition, we assume a fixed point arrival rate $v$, i.e., $t_{i+1}-t_i$ is equal for any $i$ and $v=\frac{1}{t_{i+1}-t_i}$. By referring to \cite{Cao2006Density,Chen2007Density}, given a decay model with parameters $a$ and $\lambda$, the sum of all data points' freshness $a^{\lambda(t-t_i)}$ for an unbounded data stream is a constant $\frac{v}{1-a^\lambda}$, i.e., $\sum_{i=1}^{n}\big(a^{\lambda(t_n-t_i)}\big)=\frac{v}{1-a^\lambda}$ where $n\rightarrow\infty$. Accordingly, we distinguish active and inactive cluster-cells as follows. A cluster-cell $c$ at time $t$ is active if $\rho^t_c \geq\frac{\beta\cdot{v}}{1-a^\lambda}$ and otherwise inactive. $\beta$ is a tunable parameter that controls the threshold. The larger the value of $\beta$, the less number of active cluster-cells is. Obviously, $\beta$ is less than 1 since a single cluster-cell's density should not exceed the sum of all cluster-cells' densities (which is equal to the sum of all points' freshness), i.e.,  $\frac{\beta\cdot{v}}{1-a^\lambda}\leq\rho^t_c\leq\frac{v}{1-a^\lambda}$. On the other hand, since a new cluster-cell formed by a new arrival point should be considered as inactive, we have $\rho^t_c=1<\frac{\beta\cdot{v}}{1-a^\lambda}$. Thus, we have the range of $\beta$, i.e., $\frac{1-a^\lambda}{v}<\beta<1$.


The active cluster-cell may become inactive and be moved from the DP-Tree to the outlier reservoir. Suppose a cluster-cell $c$ becomes inactive at time $t$, i.e., $\rho_c^t<\frac{\beta\cdot{v}}{1-a^\lambda}$. Due to the fact that the density of cluster-cell $c$'s successors are all lower than $\rho_c^t$, cluster-cell $c$'s successor cluster-cells should also be moved to the outlier reservoir. It is unnecessary to judge their densities or update their dependent distances. On the other hand, the inactive cluster-cells may absorb new arrival points to increase density. Then they may become active cluster-cells and be inserted into the DP-Tree. The DP-Tree insertion leads to the dependencies update. We follow Theorem \ref{theo:upnum} and Theorem \ref{theo:upcell} to reduce the overhead of dependencies update.



\vspace{-0.05in}
\subsection{Memory Space Recycling}\label{subsec:remove}

As new data are continuously being collected, more and more cluster-cells could be created due to the expansion of data space. The maintenance of cluster-cells consumes large memory space. In practice, if data are old enough they can be ignored for clustering since they are invalid for discovering the hidden patterns and the trends in stream. For the sake of recycling memory space, we delete the outdated cluster-cells. If the time of inactive cluster-cell has not absorbed any point is equal to the time for a new active cluster-cell being formed by new arrival points, the inactive cluster-cell can be deleted safely. We call this kind of inactive cluster-cell as \textit{outdated cluster-cells}. We study the time for safely deleting inactive cluster-cells through the following theorem.
\vspace{-0.05in}
\begin{theorem}\label{theo:outdatedtime}
Suppose the speed of stream is fixed as ${v}$. We can safely delete an inactive cluster-cell without any negative impact if an inactive cluster-cell has not absorbed any point for time $\Delta T_{del}$, where
\vspace{-0.05in}
\begin{equation}
\Delta T_{del} > \frac{log_a(1-a^\lambda)-log_a(\beta\cdot v)}{\lambda\cdot v}.
\end{equation}
\end{theorem}
\vspace{-0.05in}

\begin{theorem}\label{lem:outdatedtime}
Suppose the speed of stream is fixed as ${v}$. We can safely delete an inactive cluster-cell without any negative impact if an inactive cluster-cell has not absorbed any point for time $\Delta T_{del}$, where
\begin{equation}
\Delta T_{del} > \frac{log_a(1-a^\lambda)-log_a(\beta\cdot v)}{\lambda\cdot v}.
\end{equation}
\end{theorem}
\begin{proof}

The density of an existing inactive cluster-cell $c$ is $\rho$. Suppose after a time interval $\Delta T$, it become large enough to be active. The smaller the value of $\rho$ is, the longer time $\Delta T$ is. Since we want to figure out the maximum $\rho$ such that $c$ could be safely deleted, we should consider the case when $\Delta T$ is the minimum to let $c$ become active. It is obvious that $\Delta T$ is the minimum when all the points are absorbed by $c$. Suppose there is a $\rho$ such that, at time $\Delta T$ the density of $c$ exceeds the active threshold, but at time $\Delta T-\Delta t$ the density of $c$ does not exceed the active threshold, where $\Delta t=1/v$ is the time interval between two continuously arrived points. The density of $c$ at time $\Delta T$ is
\begin{equation}
\label{eq:eqeq1}
\rho\cdot a^{\lambda\cdot v\cdot\Delta T}+\frac{1-a^{\lambda\cdot v\cdot\Delta T}}{1-a^\lambda}\geq\frac{\beta\cdot v}{1-a^\lambda}
\end{equation}
where the former part is the density contribution before and the later part is the new density contribution. Recall that the existing inactive cluster-cell is always denser than the newly generated one. The density of the newly generated one is smaller than $c$ and also smaller than the active threshold at time $\Delta T-\Delta t$.
\begin{equation}
\label{eq:eqeq2}
\frac{1-a^{\lambda\cdot v\cdot(\Delta T-\Delta t)}}{1-a^\lambda}<\frac{\beta\cdot v}{1-a^\lambda}
\end{equation}
By applying (\ref{eq:eqeq2}) to (\ref{eq:eqeq1}), we will finally have $\rho<1$. This is to be expected. This is because that the density can be increased at most 1 from time $\Delta T-\Delta t$ to $\Delta T$. If the previous density is larger than 1, the density of the existing inactive cluster-cell must be higher than that of the newly generated cluster-cell.

We then analyze how long it takes to decrease its density to be less than 1 if now new points are absorbed. The extreme case is when a newly decayed cluster-cell who's density is almost $\frac{\beta\cdot v}{1-a^\lambda}$. Suppose the time is $\Delta T_{del}$, we have
\begin{equation}
\frac{\beta\cdot v}{1-a^{\lambda}}\cdot a^{\lambda\cdot v\cdot\Delta T_{del}}<1.
\end{equation}
We then obtain $\Delta T_{del}>\frac{log_a(1-a^\lambda)-log_a(\beta\cdot v)}{\lambda\cdot v}$



\end{proof}

We analyze the theoretical upper bound of the outlier reservoir size in the following. As we have discussed in Sec. \ref{sec:overview}, the inactive cluster-cells might be decayed from DP-Tree or generated from new arrival points. Suppose that every new arrival point in a time interval $\Delta T_{del}$ forms a new inactive cluster-cell, which is the case that results in the maximum number of inactive cluster-cells. Since after time $\Delta T_{del}$ all the previous inactive cluster-cells are removed, the outlier reservoir holds all newly generated inactive cluster-cells, which is $\Delta T_{del} \cdot v$. In addition, the active cluster-cells could decay to inactive ones. Considering that the sum of all cluster-cells' densities is $\frac{v}{1-a^\lambda}$ and the active cluster-cell's density is at least $\frac{\beta\cdot{v}}{1-a^\lambda}$, the total number of active cluster-cells is at most $(\frac{v}{1-a^\lambda})/(\frac{\beta\cdot{v}}{1-a^\lambda})=\frac{1}{\beta}$. Therefore, the outlier reservoir can hold at most $\Delta T_{del} \cdot v +\frac{1}{\beta}$ inactive cluster-cells.

\section{Adaptive Tuning of \Large$\tau$}
\label{sec:adaptive}

The parameter $\tau$ in \emph{EDMStream} controls the degree of cluster separation and cluster granularity. Large $\tau$ tends to result in a less number of large clusters, while small $\tau$ tends to result in much more small clusters. As the data distribution of stream evolves over time, the key parameter $\tau$ should not be set statically but dynamically to adapt to data evolution. When the stream points are loosely distributed, the $\tau$ should be a larger value, and vice verse. The original DP Clustering \cite{rodriguez2014clustering} draws a decision graph (see Fig. \ref{dpb}) to help users determine an appropriate $\tau$ value. However, it is not suitable for stream clustering since it is performed frequently rather than once, which is expensive and causes great inconvenience to users. Even though the user-interaction method does not work, it inspires us to learn the preference of users and propose an automatic tuning approach for $\tau$. Therefore, an adaptive approach for dynamically adjusting $\tau$ is desired.

A common optimization objective in clustering is minimizing the intra-cluster distance and maximizing the inter-cluster distance. Similarly, in DP clustering we aim to minimize the average of relative intra-dependent-distance
$\frac{\sum_{c:\delta_c\leq\tau}\delta_c/\overline\delta}{m}$ and at the same time maximize the average of relative inter-dependent-distance $\frac{\sum_{c:\delta_c>\tau}\delta_c/\overline\delta}{n}$ where $m$ is the number of intra-cluster-cells $m=|\{c:\delta_c\leq\tau\}|$, $n$ is the number of inter-cluster-cells $n=|\{c:\delta_c>\tau\}|$, $\overline\delta$ is the average dependent distance $\overline\delta=\sum_{c}\frac{\delta_c}{n+m}$. With regard to time information, we propose the following evaluation function for stream clustering and aim to minimize it \footnote{It is noticeable that the outliers or halos with small $\rho$ are already excluded for evaluating the objective function in order to reduce noise.}.
\vspace{-0.05in}
\begin{equation}\label{eq:tau}
    \mathcal{F}(\tau^t)=\alpha\cdot\frac{\sum\limits_{c:\delta_c^t>\tau^t}\delta_c^t}{n\cdot\overline\delta}+(1-\alpha)\cdot\frac{m\cdot\overline\delta}{\sum\limits_{c:\delta_c^t\leq\tau^t}\delta_c^t},
\vspace{-0.05in}
\end{equation}
where $0<\alpha<1$ is a balancing parameter that reflects the preference to minimize the average intra-dependent-distance or to maximize the average inter-dependent-distance. $\alpha$ also implicitly reflects the user preference to cluster granularity.

If $\tau^t$ is a larger value, the average intra- and inter-dependent-distances are larger values, so that $\frac{m\cdot\overline\delta}{\sum_{c:\delta_c^t\leq\tau^t}\delta_c^t}$ become extremely large, and we will obtain less number of large clusters. While $\tau^t$ is a smaller value, the average intra- and inter-dependent-distances are smaller values, so that $\frac{\sum_{c:\delta_c^t>\tau^t}\delta_c^t}{n\cdot\overline\delta}$ become extremely large, and it will result in a large number of small clusters. Therefore, a proper $\tau$ should tend to minimize $\mathcal{F}(\tau^t)$.


Another problem is how to choose $\alpha$. We choose $\alpha$ by learning user preference based on their initial selection of cluster centers from decision graph. We adopt the following heuristic to estimate $\alpha$. In the initialization phase after a number of cluster-cells are cached and form an initial DP-Tree, we draw a decision graph according to these existing cluster-cells' $\rho$ and $\delta$ values and let users pick cluster centers. Suppose user chooses a number of density peaks whose dependent distances are at least $\tau^0$.
Given the initial $\tau^0$, we can find the $\alpha=\widehat{a}$ such that for any $\delta \neq \tau^0$, $\mathcal{F}(\widehat{a},\tau^0)<\mathcal{F}(\widehat{a},\delta)$. Given the balance parameter $\alpha$ which reflects user preference, at any time $t$ a $\tau^t$ that minimizes $\mathcal{F}(\tau^t)$ can be determined automatically.

\section{experimental evaluation}
\label{sec:expr}
In this section, we present the experimental evaluation of \textit{EDMStream}. All experiments are conducted on a commodity PC with 3.4GHz Intel Core i3 CPU and 8GB memory.
\vspace{-0.05in}
\subsection{Preparation}
\vspace{-0.05in}
\Paragraph{Datasets}
Our experiments involve two synthetic datasets (SDS and HDS) and four real datasets (NADS \cite{Lichman:2013}, KDDCUP99 \cite{Stolfo1999Cost}, CoverType \cite{Malik1999Estimating, Reiss2012Introducing} and PAMAP2 \cite{Reiss2012Creating}). All the real datasets are with ground truth information. The SDS dataset \cite{pei2006synthetic} is generated with 2-D stream points to visually show cluster results. The HDS stream dataset is generated with various dimensions based on the approach mentioned in \cite{vennam2005syndeca}. Both the synthetic and real datasets are converted into streams by taking the data input order as the order of streaming except for NADS. The NADS is a news stream dataset with time information (there is no dimension information for NADS since the data are short text). The features of these datasets are listed in Table \ref{tab:dataset}.

\Paragraph{Comparison Counterparts}
We implement D-Stream \cite{Chen2007Density}, DenStream \cite{Cao2006Density}, DBSTREAM \cite{hahsler2016clustering} and MR-Stream \cite{Wan2009Density} for comparison\footnote{The source codes of these algorithms are not public available except for DBSTREAM.}, all of which are density-based. The counterparts all use an online component (i.e., data summarization structure) that maps each stream object to a grid \cite{Chen2007Density, Wan2009Density} or microcluster \cite{Cao2006Density, hahsler2016clustering} 
and an offline component that performs a batch-mode classical clustering algorithm to update clustering result.

\Paragraph{Parameters Setup}
In our experiment, we fix the data points arrival rate as 1,000 pt/s unless particularly mentioned. Different algorithms have different decay parameter requirements. In order to make equal decay effect, we carefully set the decay parameters as follows. In \emph{EDMStream} and D-Stream, we set $a=0.998$ and $\lambda = 1$ such that $a^\lambda=0.998$. In MR-Stream, because $a=1.002$ is fixed, so we set $\lambda=-1$ such that $a^\lambda=0.998$. In DenStream and DBSTREAM, they fix $a=2$, therefore we set $\lambda=0.0028$ such that $a^\lambda = 0.998$. We use these parameter settings for achieving the same decay rate. In addition, we set $\beta = 0.0021$ as discussed in Sec. \ref{sec:stream:emergedecay}. For the radius $r$, we refer to the method of choosing $d_c$ in \cite{rodriguez2014clustering}. We will discuss the effect of $r$ in Sec. \ref{subsec:radius}. We set the other parameters of the competitor algorithms by referring to their papers.
\vspace{-0.2in}
\begin{center}
\begin{table}[htb]
    \caption{Datasets}
    \vspace{-0.1in}
    \label{tab:dataset}
    \centering
    \footnotesize
    \begin{tabular}{c|c|c|c|c}
    \hline
    \textbf{data set}&\textbf{instances}& \textbf{dim} &\textbf{clusters}&\textbf{$r$}\\
\hline
SDS & 20,000 & 2 & 2 &  0.3 \\
\hline
\multirow{4}{*}{HDS} & \multirow{5}{*}{100,000} & 10 & 20 & 60 \\
\cline{3-5}
& & $30$ & 20 & 65\\
\cline{3-5}
& & $100$ & 20 & 68\\
\cline{3-5}
& & $300$ & 20 & 70\\
\cline{3-5}
& & $1000$ & 20 & 70 \\
\hline
NADS & 422,937  & - & 7231 & 0.4\\
\hline
KDDCUP99 & 494,021 & 34 & 23 & 100 \\
\hline
CoverType & 581,012 & 54 & 7  & 250\\
\hline
PAMAP2 & 447,000 & 51 & 13 & 5\\
\hline
\end{tabular}
\vspace{-0.2in}
\end{table}
\end{center}

\vspace{-0.05in}
\subsection{Tracking Cluster Evolution}

An important feature of a stream clustering algorithm is the ability to track the evolution of clusters. We first use a synthetic 2-D dataset SDS to visually show the ability of tracking clustering evolution. We then show a use case of news recommendation application to illustrate how to utilize the ability of clustering evolution tracking.
\vspace{-0.05in}
\subsubsection{Synthetic Dataset}

Fig. \ref{6result} shows the data stream with 6 snapshots taken in $t_1=1s$, $t_2=4s$, $t_3=8s$, $t_4=12s$, $t_5=14s$, $t_6=20s$. Various degrees of grey indicate the freshness of data. The darker one is fresher, the lighter one is staler. SDS contains 20,000 points and the point arrival speed is set to 1,000 pt/s, so the stream ends at 20s.

We show the cluster evolution tracking result in Fig. \ref{fig:clusterlifeEDS}. Different color lines indicate different clusters. The length of lines indicates the lifecycle of clusters. Multiple branches split from one line or merging into one means cluster splitting or merging. According to Fig. \ref{6result}, from 1s to 4s, we can see that the shapes and locations of two clusters are evolving, they are moving closer to each other. At 9s, these two clusters merge into one single cluster. At 12s, a new cluster emerges at right hand side, and the left cluster shrinks gradually. At 14s, the left cluster disappears completely, and the right cluster has been split into two different clusters. They are also moving to the opposite directions from each other. Finally at 20s, the two clusters move far away from each other. We can see that Fig. \ref{fig:clusterlifeEDS} successfully captures all the cluster evolution activities.

\begin{figure}[t]
    \captionsetup[subfigure]{farskip = -0.08in}
    \centering
    \subfloat[$t_1=1s$]{\includegraphics[width = 1in]{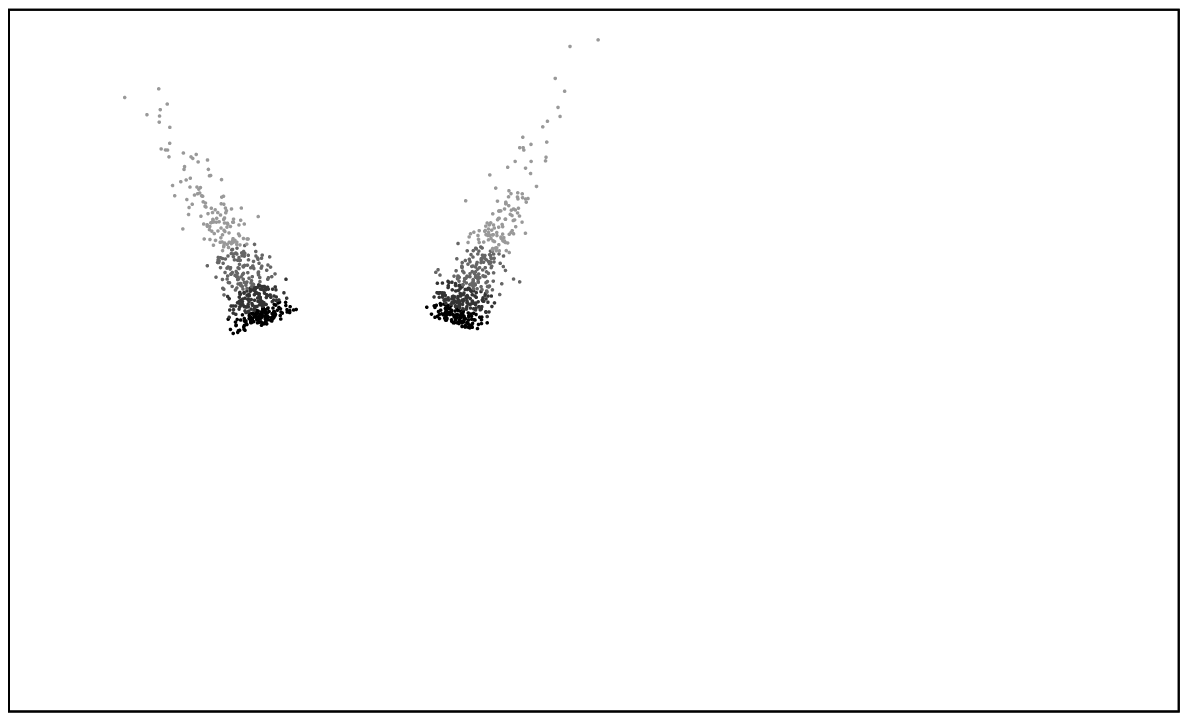}\label{t1}}\ \
    \vspace{0pt}
    \subfloat[$t_2=4s$]{\includegraphics[width = 1in]{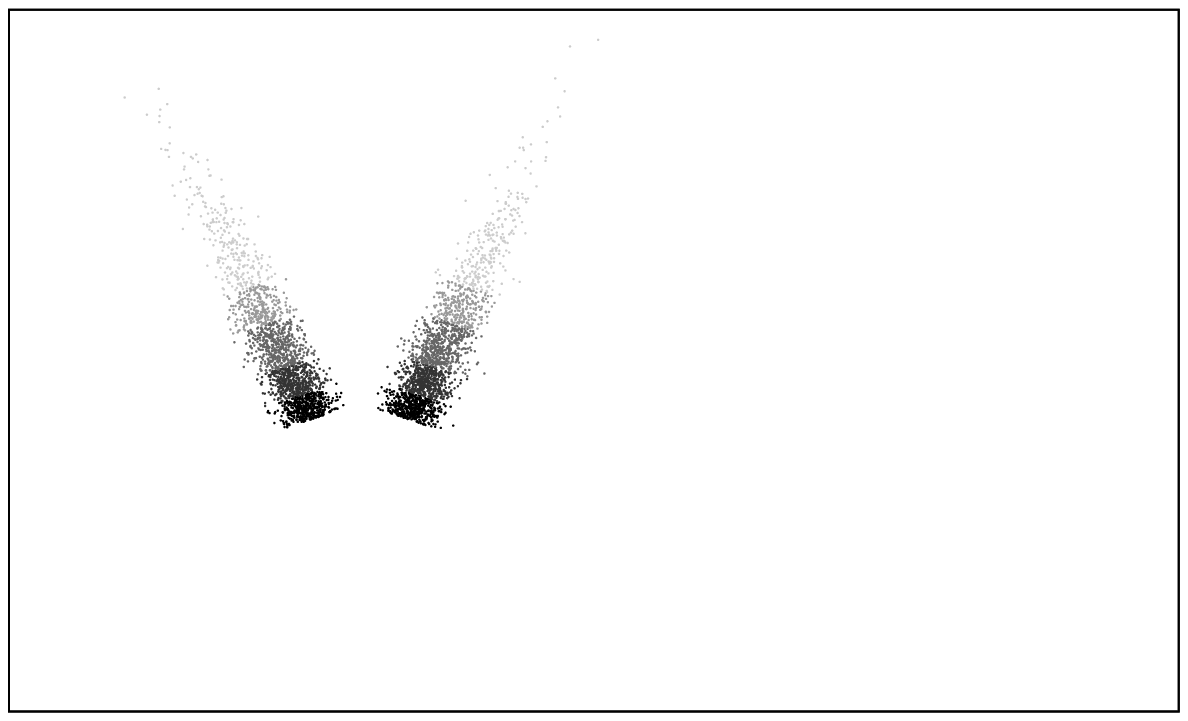}\label{t2}}\ \
    \vspace{0pt}
    \subfloat[$t_3=8s$]{\includegraphics[width = 1in]{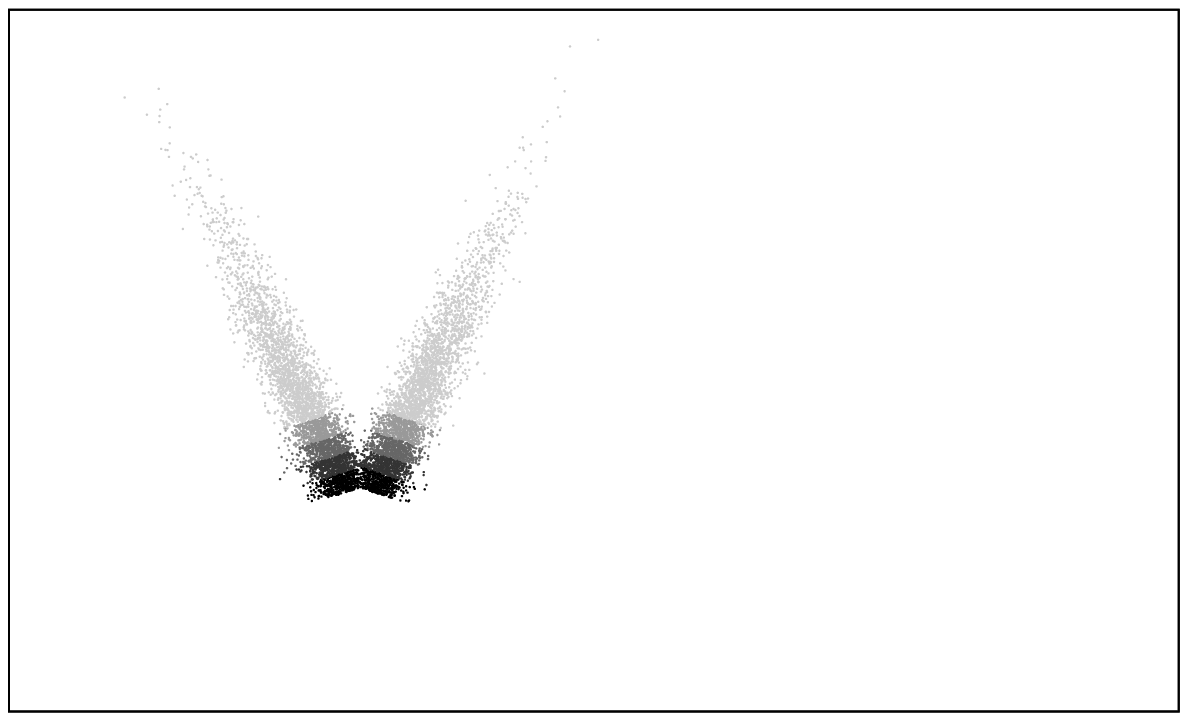}\label{t3}}\\
    \vspace{0pt}
    \subfloat[$t_4=12s$]{\includegraphics[width = 1in]{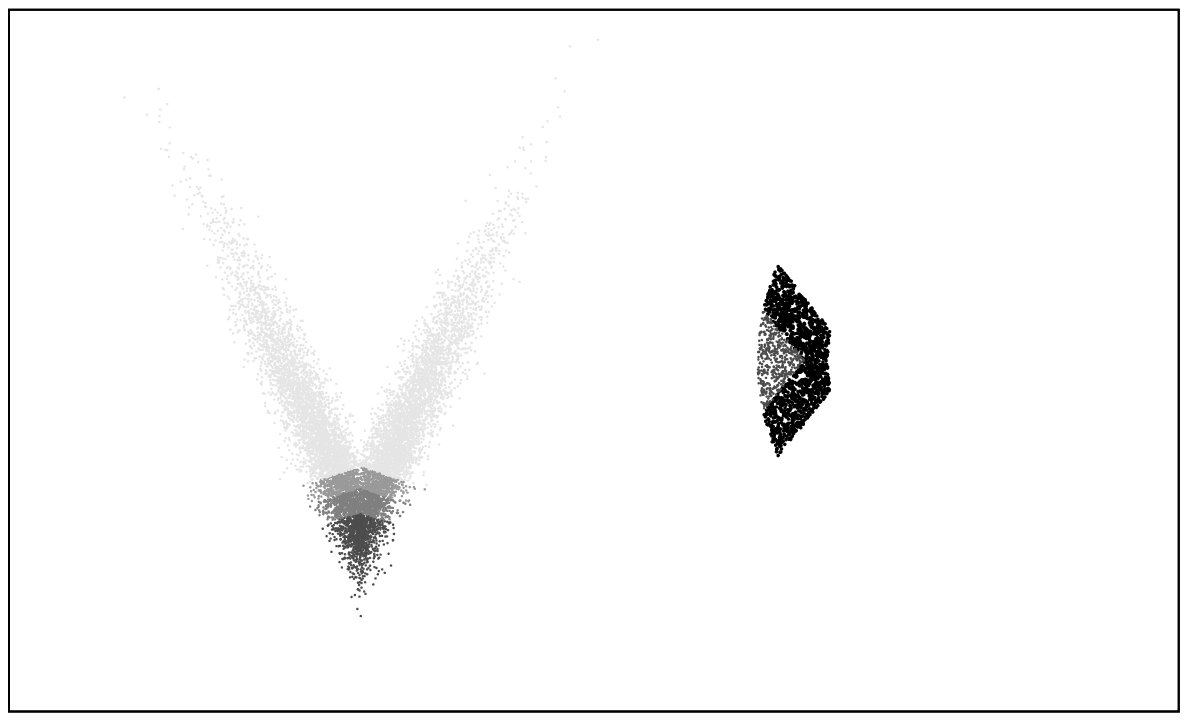}\label{t4}}\ \
    \vspace{0pt}
    \subfloat[$t_5=14s$]{\includegraphics[width = 1in]{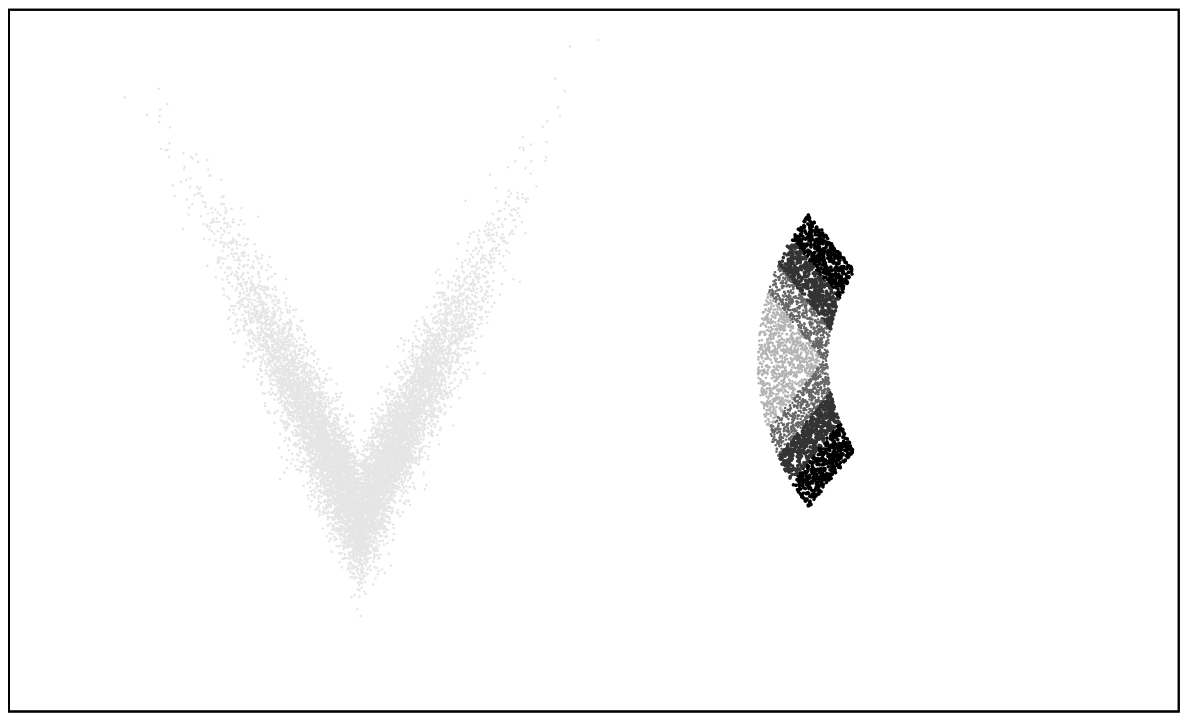}\label{t5}}\ \
    \vspace{0pt}
    \subfloat[$t_6=20s$]{\includegraphics[width = 1in]{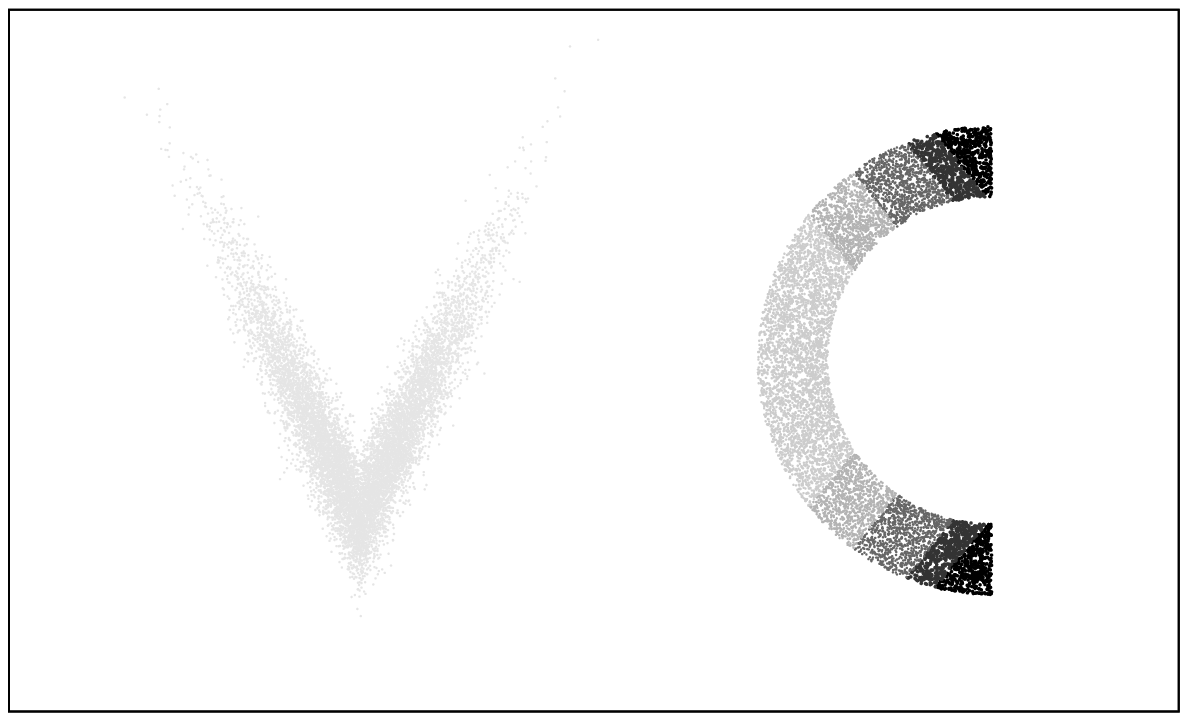}\label{t6}}\\
    \captionsetup[subfigure]{labelformat=empty}
    \vspace{6pt}
    \subfloat[]{\includegraphics[width = 3in]{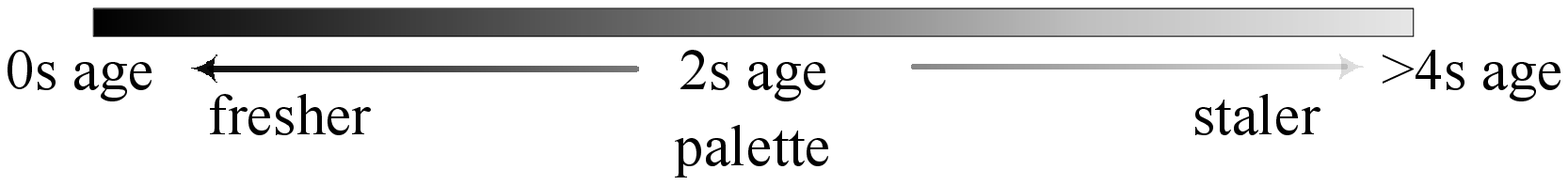}\label{palette}}
    \vspace{-0.25in}
\caption{Data distribution snapshots with time decay information at different time points (SDS)}\label{6result}
\vspace{-0.15in}
\end{figure}

\begin{figure}[t]
  \centering
  \includegraphics[width=2.5in]{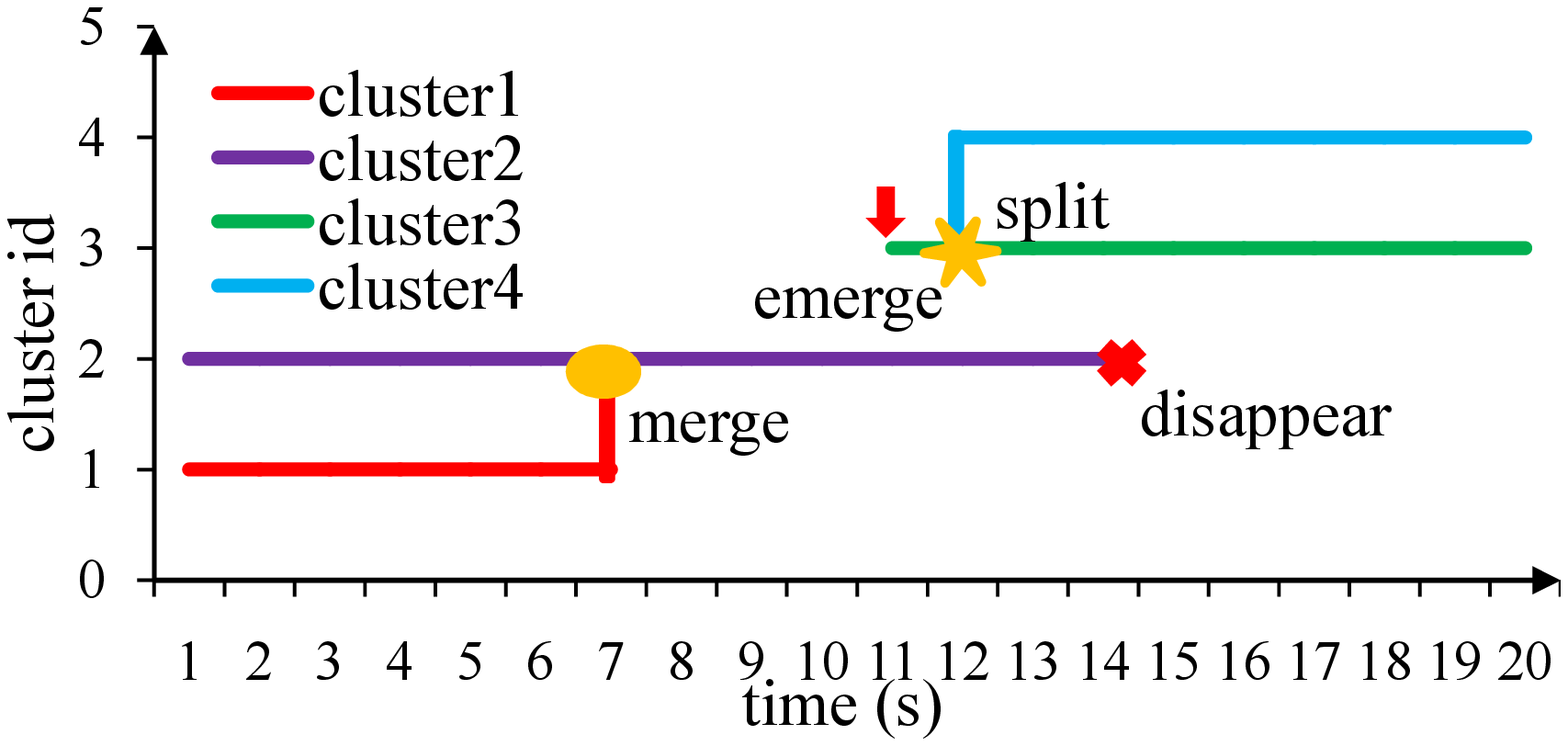}
  \vspace{-0.1in}
  \caption{Cluster evolution activities (SDS)}\label{fig:clusterlifeEDS}
  \vspace{-0.2in}
\end{figure}

\vspace{-0.05in}
\subsubsection{Use Case Study: Monitoring Cluster Evolution in News Recommendation}

\begin{figure}[t]
  \centering
  \includegraphics[width=3in]{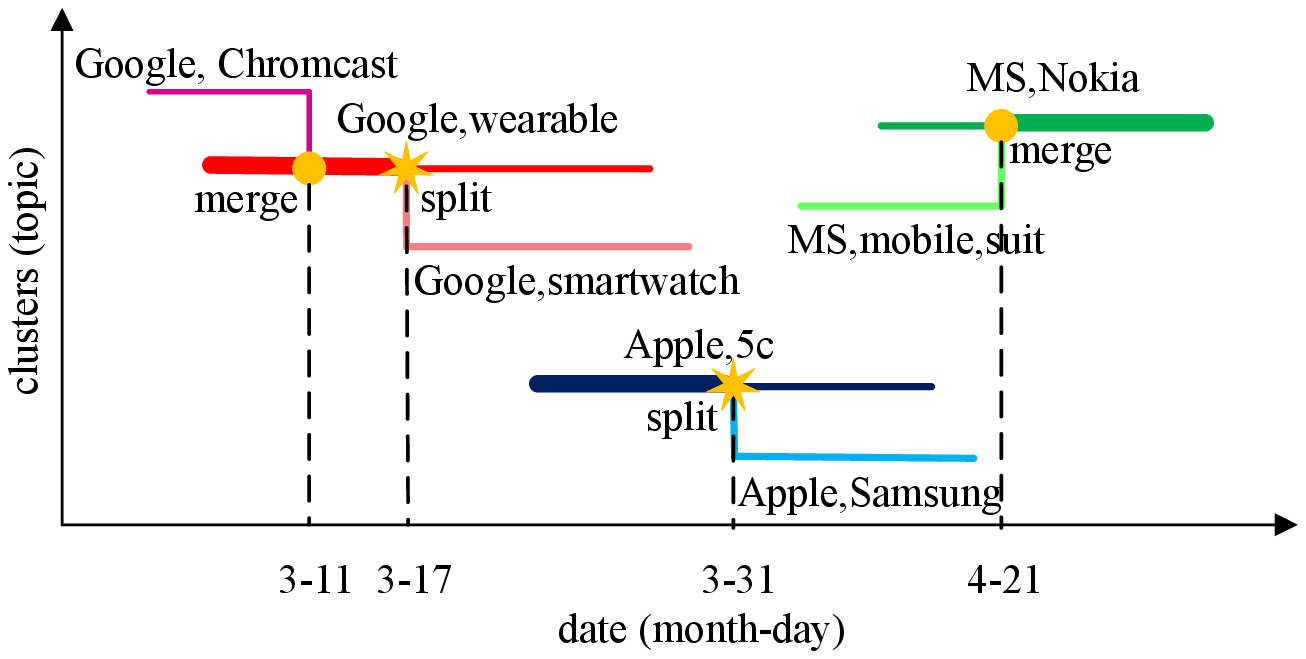}
  \vspace{-0.1in}
  \caption{Cluster evolution activities (NADS)}\label{fig:clusterlifeNews}
  \vspace{-0.2in}
\end{figure}

A real application of stream clustering is news recommendation. The news in the same cluster as that a user has visited is recommended to the user. For the text dataset, the Jaccard distance is used. The density of cluster-cell is computed as defined in Equation \ref{eq:timedensity}. We run \emph{EDMStream} on the NADS news stream and depict the cluster evolution tracking result in Fig. \ref{fig:clusterlifeNews}\footnote{There should be about 7281 clusters at the end of the stream. We only show a few cluster evolution tracking results in the figure.}. We show the tags of each news cluster which indicate the topic. Furthermore, we also show the real events which we think lead to the cluster evolution in Table \ref{tab:cluevent}.

On 3-11, a cluster with tags \{\textit{Google}, \textit{Chromcast}\} merges into another cluster with tags \{\textit{Google}, \textit{wearable}\}. We analyze the reason as follows. The news about Chromcast are not hot anymore, but these news are highly related to another news topic \{\textit{Google}, \textit{wearable}\}. Given that many news about ``Google launches SDK for Android wearables'' come out at that time, these two clusters merge together. Later, Google confirms smartwatch plans. These news firstly are classified into the \{\emph{Google}, \emph{wearables}\} cluster. But with the increase of popularity, they are split from original cluster and form a new cluster \{\emph{Google}, \emph{smartwatch}\} on 3-17. On 3-31, the news cluster \{\emph{Apple}, \emph{Samsung}\} is split from the original cluster \{\emph{Apple}, \emph{5c}\}, because there is an shocked event that Apple and Samsung battle for patent and the growing media focus on it instead of Apple's iphone 5c. Similarly, on 4-21 the cluster \{\textit{MS}, \textit{mobile}, \textit{suit}\} merges into another cluster with tags \{\textit{MS}, \textit{Nokia}\} since the news about ``Microsoft's acquisition of Nokia'' become popular and the news about ``Microsoft' mobility office suite'' get less and less attention. From this experiment, we can see that our cluster evolution tracking method can successfully identify different types of cluster evolution activities, including cluster emerging, disappearing, splitting and the merging of clusters.


\begin{center}
\vspace{-0.2in}
\begin{table*}[htb]
\vspace{-0.2in}
    \caption{Cluster evolutions and their related events}
    \vspace{-0.1in}
    \label{tab:cluevent}
    \centering
    \footnotesize
    \begin{tabular}{|c|c|c|c|}
\hline
\multicolumn{4}{|c|}{\textbf{Split}} \\
\hline
\textbf{original cluster} & \textbf{cluster 1} & \textbf{cluster 2} & \textbf{event} \\
\hline
Google wearable & Google wearable & Google smartwatch & On 3-17, ``google confirms smartwatch plans unveils android wear'' \\
\hline
Apple 5c & Apple 5c & Apple Samsung & On 3-31, ``apple samsung renew patent battle court'' \\
\hline
\hline
\multicolumn{4}{|c|}{\textbf{Merge}} \\
\hline
\textbf{cluster 1} & \textbf{cluster 2} & \textbf{merged cluster} & \textbf{event} \\
\hline
Google Chromast & Google wearable & Google wearable & On 3-11, ``google exec promises wearables sdk developers'' \\
\hline
MS mobile suit & MS Nokia & MS Nokia & On 4-21, ``msft nok nokia phones renamed microsoft mobile'' \\
\hline
\end{tabular}
\vspace{-0.2in}
\end{table*}
\end{center}

\vspace{-0.05in}
\subsection{Efficiency}

The ability of updating clustering result in real time is crucial for stream clustering. We compare \emph{EDMStream} with the competitor algorithms in terms of efficiency in the following.

\vspace{-0.01in}
\subsubsection{Response Time}
\label{sec:runtime}



\begin{figure*}[t]
    \vspace{-0.3in}
    \captionsetup[subfigure]{skip = 0pt}
    \centering
    \subfloat[KDDCUP99]{\includegraphics[width = 2in]{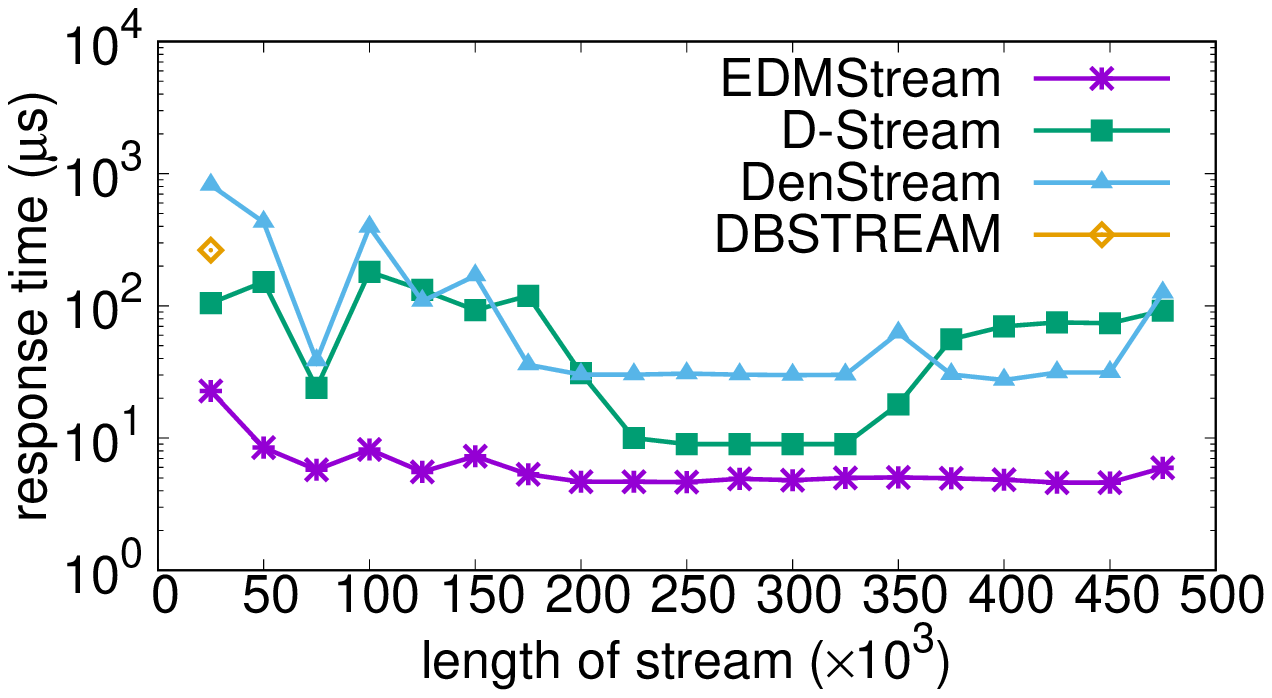}\label{fig:kddcup99runtime}}\ \
    \hspace{2pt}
    \vspace{-0.05in}
    \subfloat[CoverType]{\includegraphics[width = 2in]{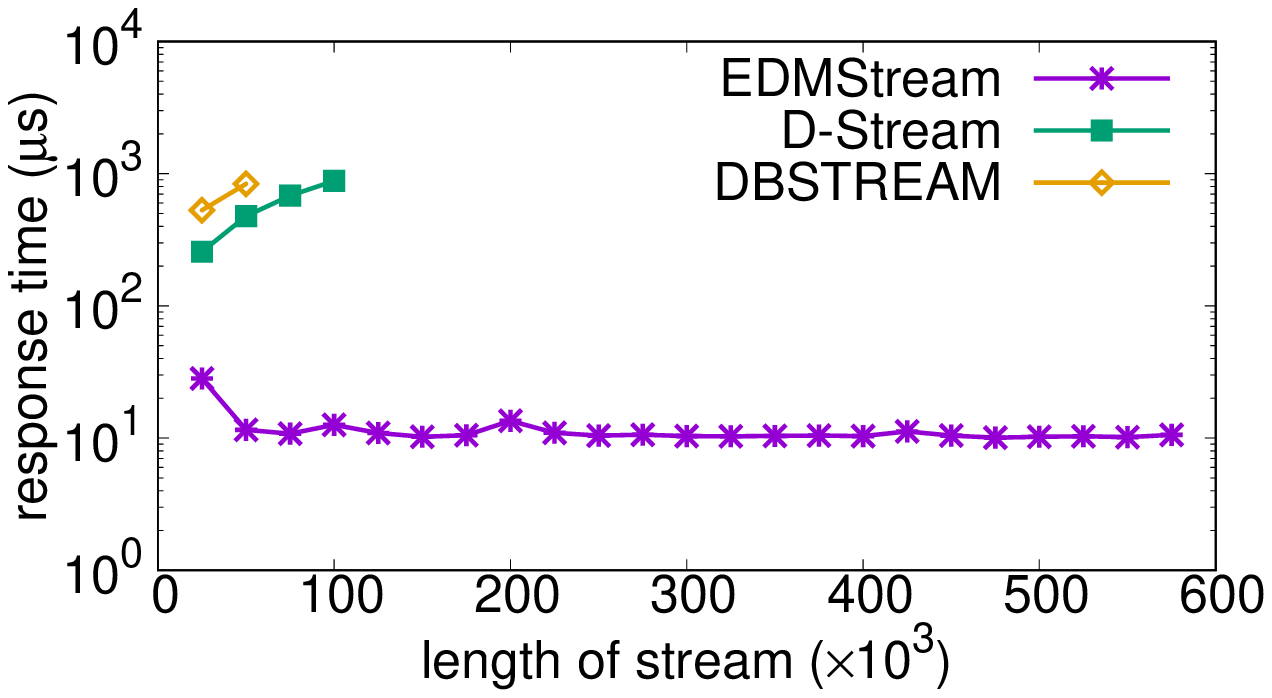}\label{fig:covtyperuntime}}\ \
    \hspace{2pt}
    \vspace{-0.05in}
    \subfloat[PAMAP2]{\includegraphics[width = 2in]{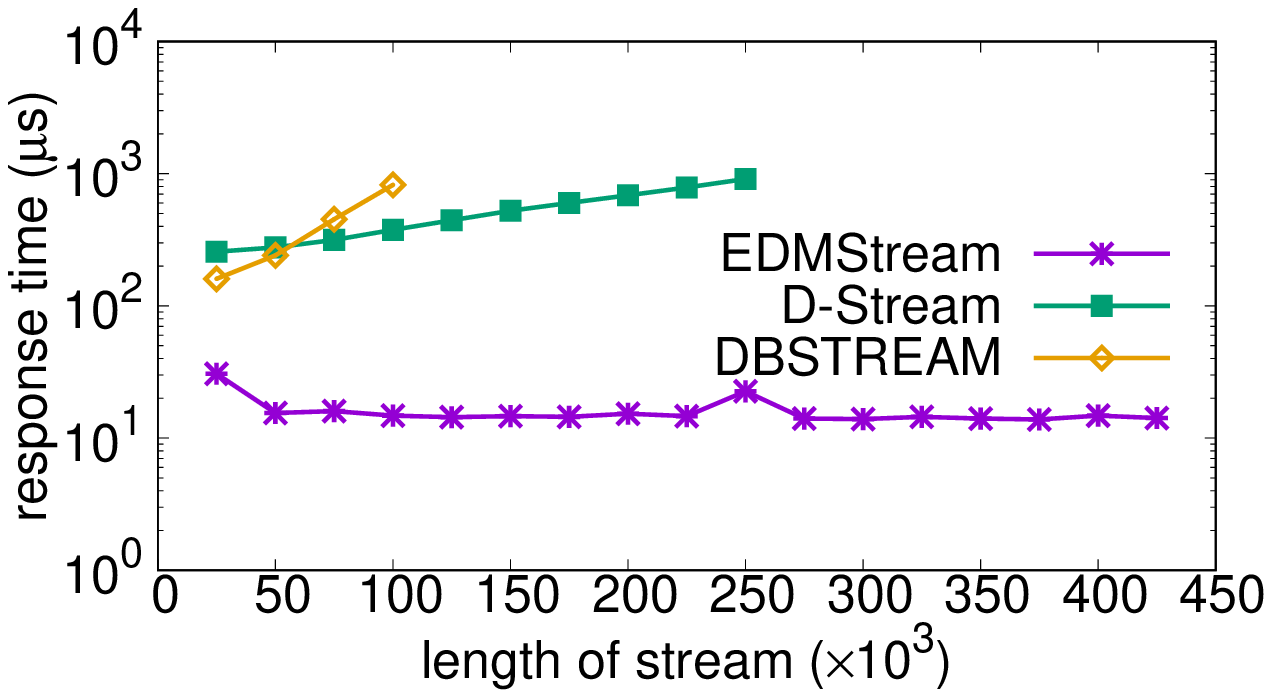}\label{fig:pparuntime}}\ \
    \vspace{-0.05in}
    \caption{The comparison of response time}\label{fig:runtime}
    \vspace{-0.2in}
\end{figure*}

We run \textit{EDMStream} and their competitor algorithms with fixed point arrival rate 1K/s. Fig. \ref{fig:runtime} shows the average response time of different algorithms in a time interval of 25 s. MR-Stream fails to process stream with 1K/s on all streams. DenStream fails on CoverType and PAMAP2 streams. DBSTREAM and D-Stream work well in the beginning but also run out of memory later. Only \emph{EDMStream} is fast enough to process stream with 1K/s point rate. \emph{EDMStream} requires much less response time than others due to the fact that \emph{EDMStream} relies on online and incremental cluster update while the others relies an costly offline clustering step.

\vspace{-0.05in}
\subsubsection{Throughput}\label{sec:throughput}

\begin{figure*}[hbt]
\vspace{-0.5in}
    \centering
    \subfloat[KDDCUP99]{\includegraphics[width = 2in]{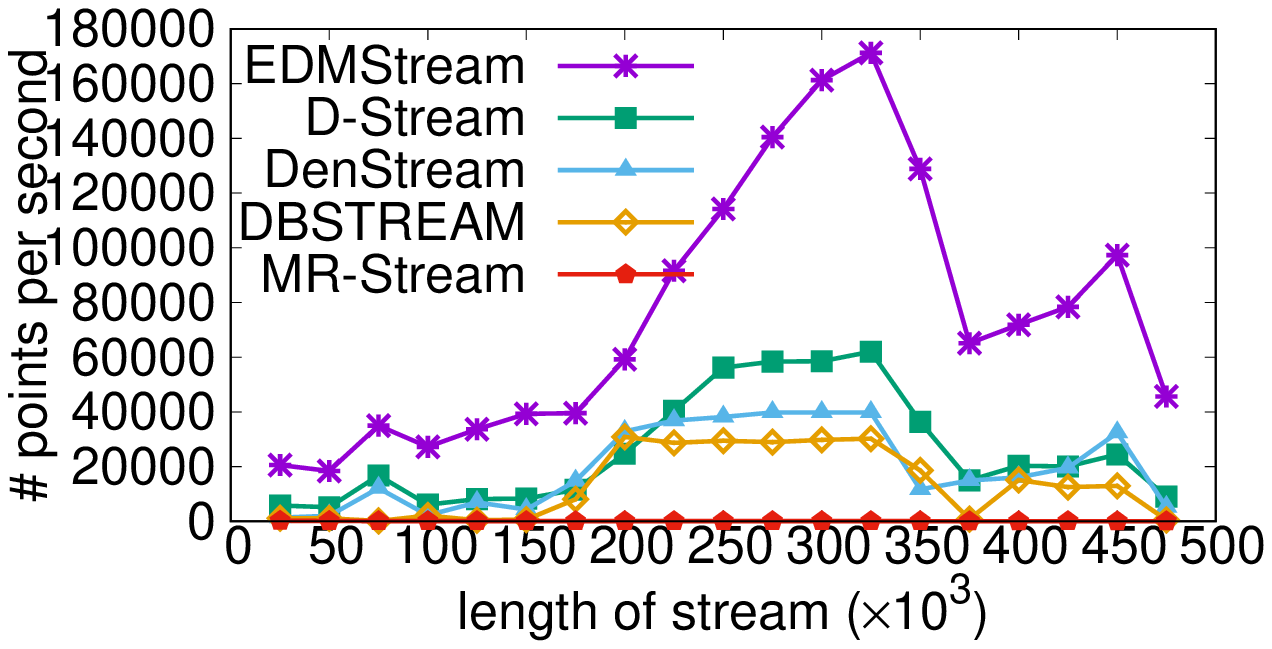}\label{fig:kddcup99thrput}}\ \
    \hspace{2pt}
    \vspace{0pt}
    \subfloat[CoverType]{\includegraphics[width = 2in]{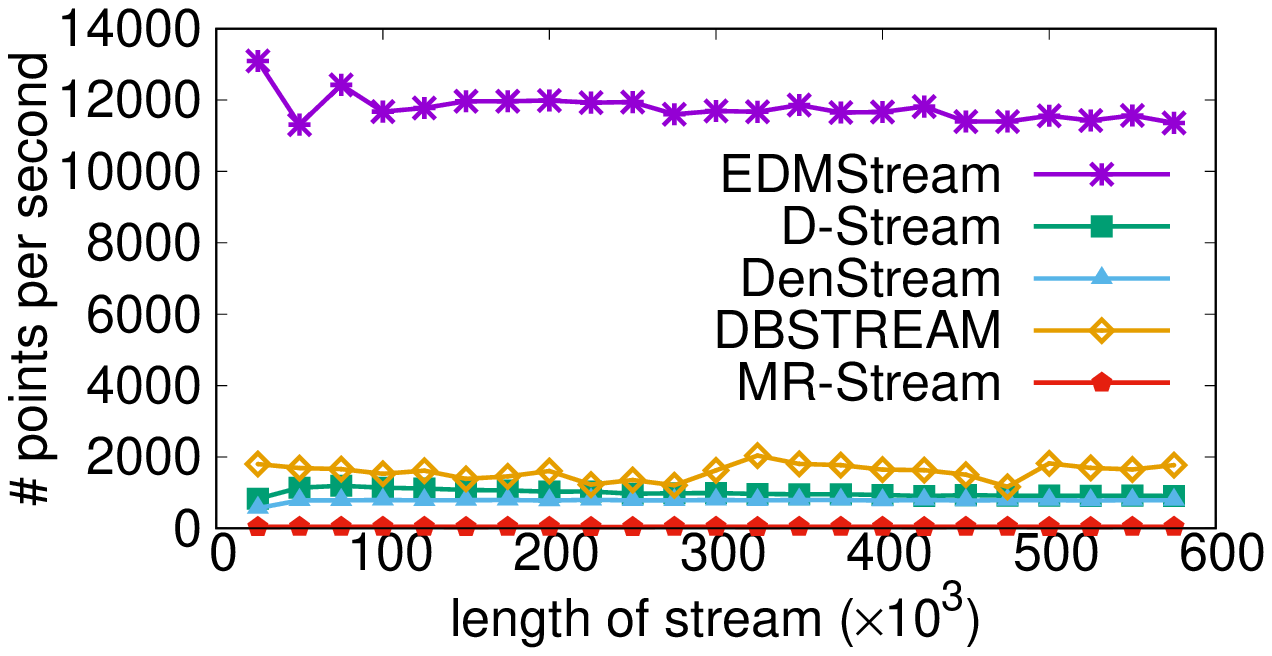}\label{fig:covtypethrput}}\ \
    \hspace{2pt}
    \vspace{0pt}
    \subfloat[PAMAP2]{\includegraphics[width = 2in]{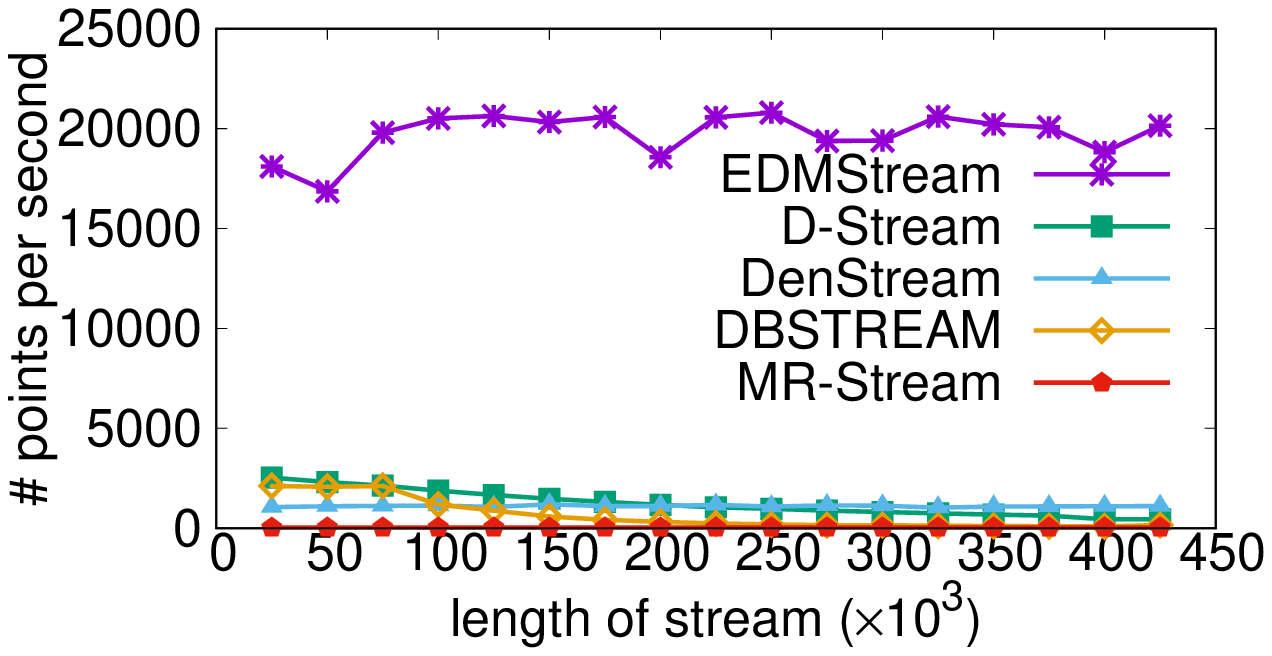}\label{fig:ppathrput}}\ \
    \vspace{-0.1in}
    \caption{Comparison of throughput}\label{fig:throughput}
    \vspace{-0.1in}
\end{figure*}

We run a stress test to see the maximum throughput that we can achieve. We remove the limitation of 1K/s point arrival rate and process as many points as possible. Fig. \ref{fig:throughput} shows the results of our \emph{EDMStream} and its competitors. The throughput of \emph{EDMStream} can be up to 10K-170K points per second, which achieves 7-15x speedups than other algorithms.


\vspace{-0.05in}
\subsubsection{Filtering Strategies}\label{sec:uptime}

The clustering result update in \textit{EDMStream} requires to update all the dependent distances ($\delta^t$) of all cluster-cells, which is the most costly step. We propose two filtering strategies to avoid unnecessary updates (see Sec. \ref{sec:stream:update}). To illustrate the effect of these filtering strategies, we first run \textit{EDMStream} without any filtering (i.e., wf). We then turn on the density filtering scheme described in Lemma \ref{theo:upnum} (i.e., df) and the triangle inequality filtering scheme described Lemma \ref{theo:upcell} (i.e., df + tif). The accumulated time for the dependencies update is depicted in Fig. \ref{fig:uptime}. We can see that our filtering schemes greatly reduce the update time.

\begin{figure*}[t]
\vspace{-0.3in}
    \captionsetup[subfigure]{skip = -0.05in}
    \centering
    \subfloat[KDDCUP99]{\includegraphics[width = 2in]{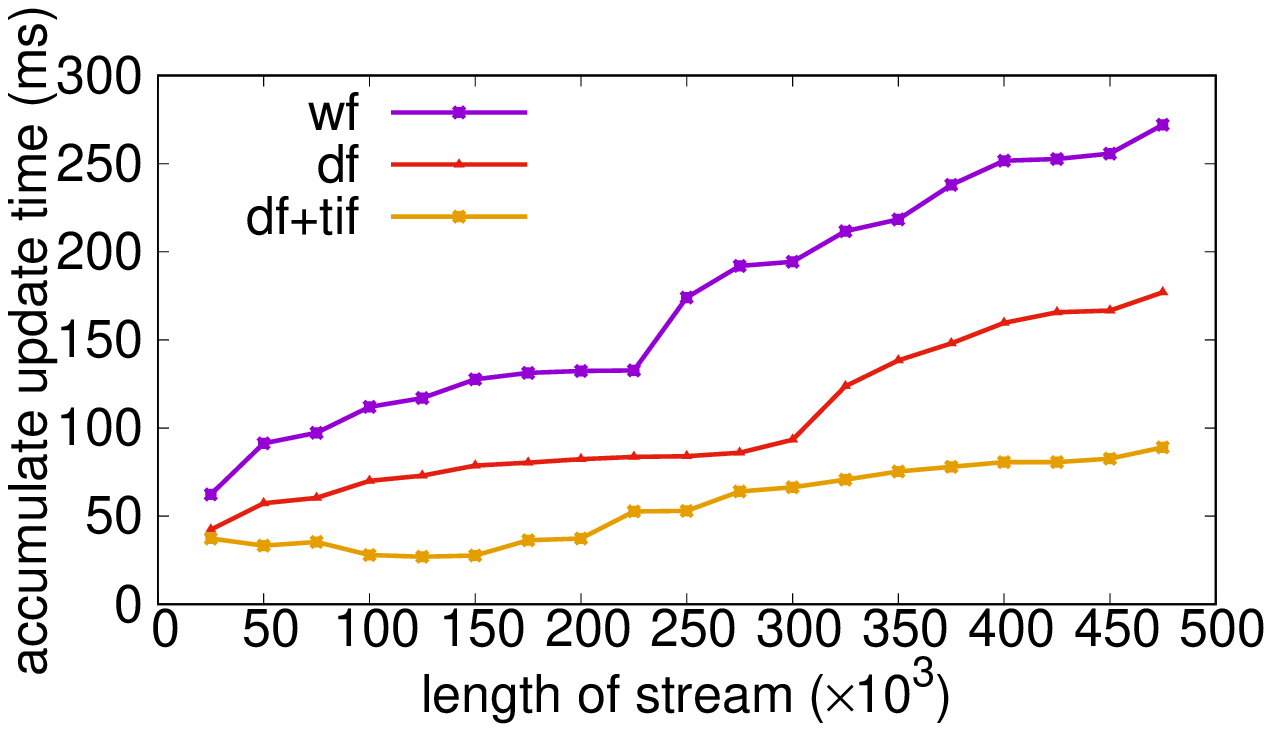}\label{fig:kdd99uptime}}\ \
    \hspace{2pt}
    \subfloat[CoverType]{\includegraphics[width = 2in]{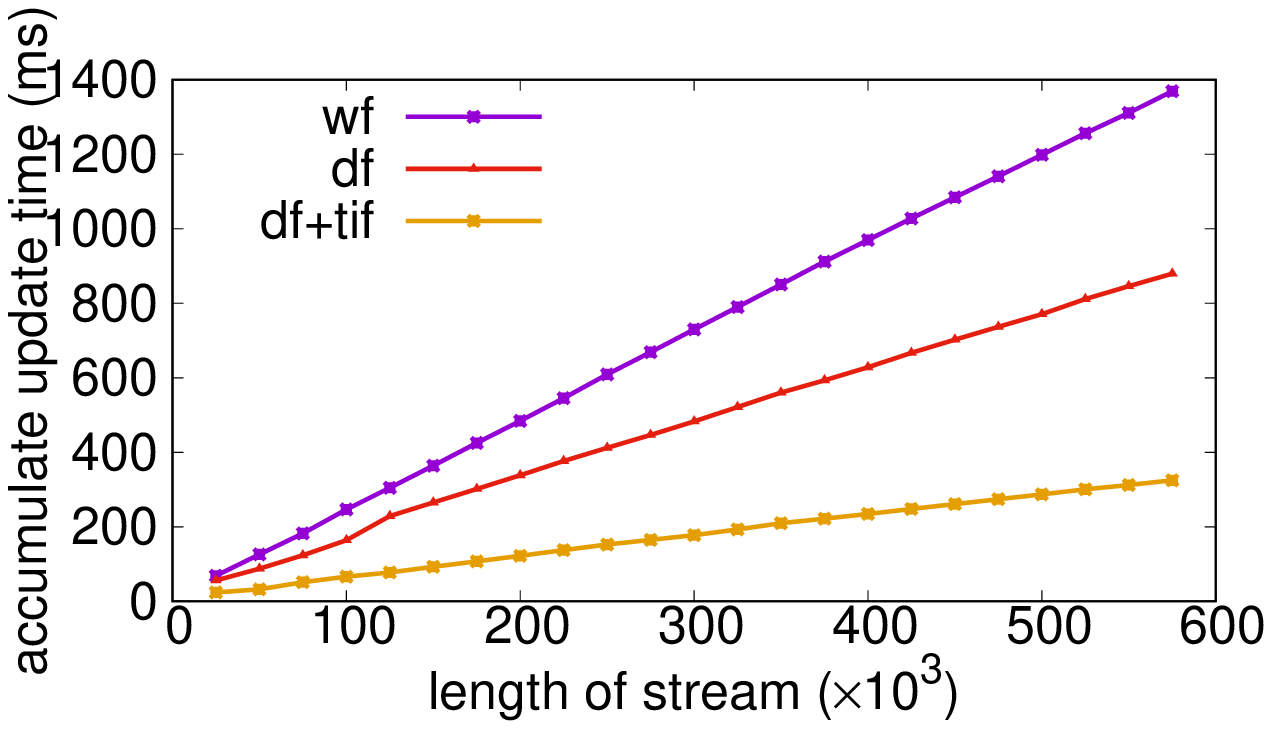}\label{fig:covtypeuptime}}\ \
    \hspace{2pt}
    \subfloat[PAMAP2]{\includegraphics[width = 2in]{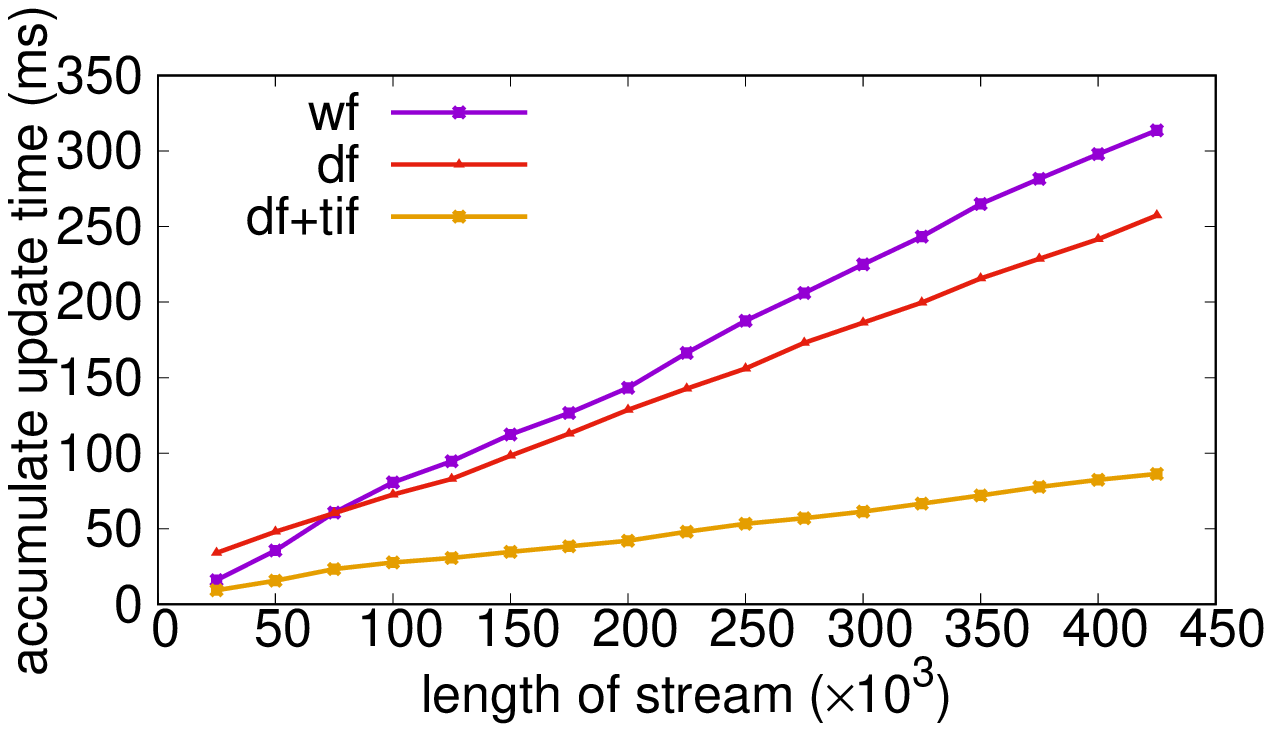}\label{fig:ppauptime}}\ \
    \vspace{-0.08in}
    \caption{Accumulated time for dependencies update}\label{fig:uptime}
    \vspace{-0.2in}
\end{figure*}

\vspace{-0.05in}
\subsubsection{Varying Data Dimensions}\label{sec:dimensition}

\begin{figure}[t]
\vspace{-0.2in}
    \centering
    \includegraphics[width=2.2in, height=1.5in]{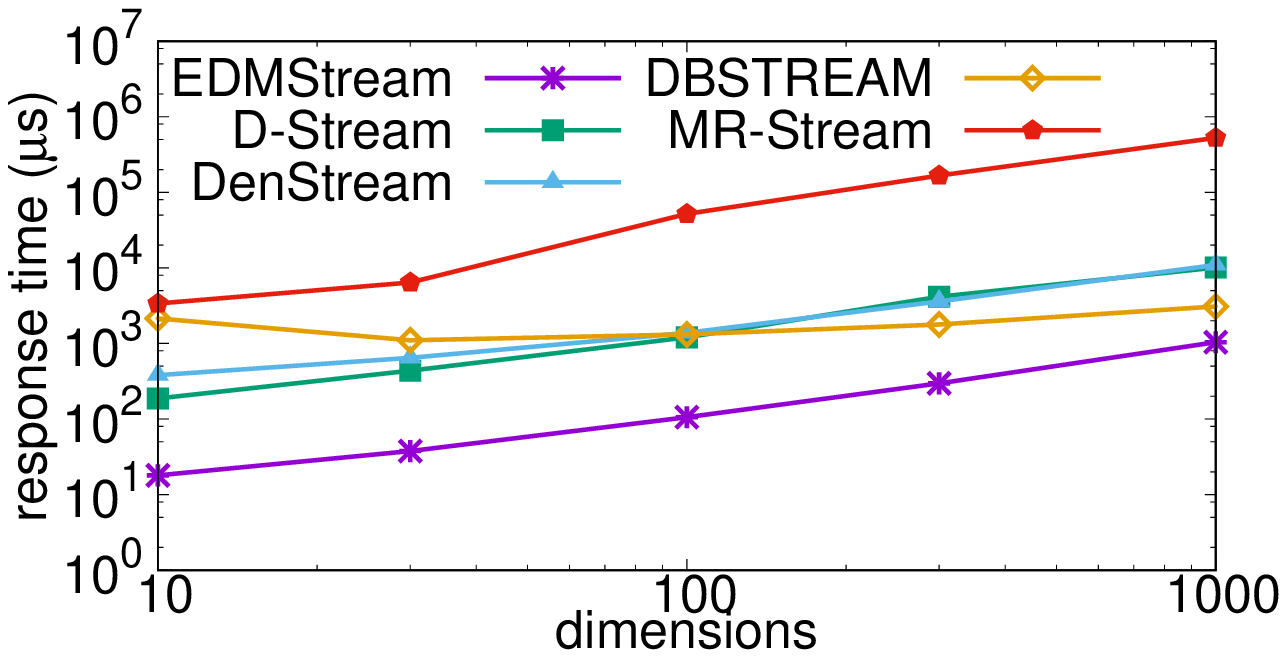}
    \vspace{-0.25in}
    \caption{Response time while varying data dimensions (HDS)}\label{fig:dimension}
    \vspace{-0.25in}
\end{figure}

We evaluate the \emph{EDMStream} on synthetic streams HDS with different data dimensions, 10D, 30D, 100D, 300D, 1000D. In this experiment, we remove the speed limits for preventing some competitors to fail. Fig. \ref{fig:dimension} shows the average response time of \textit{EDMStream} and its competitors on these streams. These algorithms exhibit similar trend when processing various-dimensions data streams. As the dimensionality increases, most algorithms require more time to update clustering result. This is under expectation, since more computation cost is needed for high dimensional data. The reason why DBSTREAM shows abnormal result is that, the performance of DBSTREAM is sensitive to the density of space and it runs faster on low density space. A stream with low dimensionality could lead to relatively high density of space, so there is a tradeoff between extra cost for high density data and the computation cost for high dimensional data.

\vspace{-0.05in}
\subsection{Cluster Quality}\label{sec:EDMandD}

To test the clustering result quality, we take state-of-the-art algorithms D-Stream, DenStream, DBSTREAM and MR-Stream as our competitors. Furthermore, we also evaluate the clustering results quality by varying stream rates.

The commonly used cluster quality evaluation metrics fail to take the evolution of streams and the freshness of instances into account. Take the \textit{F-measure} \cite{org1993Introduction} as an example. The merge or split of clusters retrieved from stream points could be considered as false positive or false negative. In our evaluation, we use a recently proposed \emph{CMM} (\emph{Clustering Mapping Measure}) criterion \cite{Kremer2011An}, which is external criterion taking into account the age of objects. The CMM measure can reflect errors related to emerging, splitting, or moving clusters, which are situations inherent to the streaming context. The CMM measure is a combination of penalties for each one of the following faults: 1) Missed objects $―$ clusters that are constantly moving may eventually “lose” objects, and thus CMM penalizes for these missed objects; 2) Misplaced objects $―$ clusters may eventually overlap over the course of the stream, and thus CMM penalizes for misplaced objects; 3) Noise inclusion $―$ CMM penalizes for noisy objects being inserted into existing clusters. Basically, the CMM value ranges from 0 to 1. Larger CMM value indicates higher cluster quality, while smaller CMM value indicates lower cluster quality.
\vspace{-0.05in}
\subsubsection{Compare with state-of-the-art algorithms}\label{sec:cmmcmpare}

\begin{figure*}[t]
\vspace{-0.4in}
    \centering
    \subfloat[KDDCUP99]{\includegraphics[width = 2in]{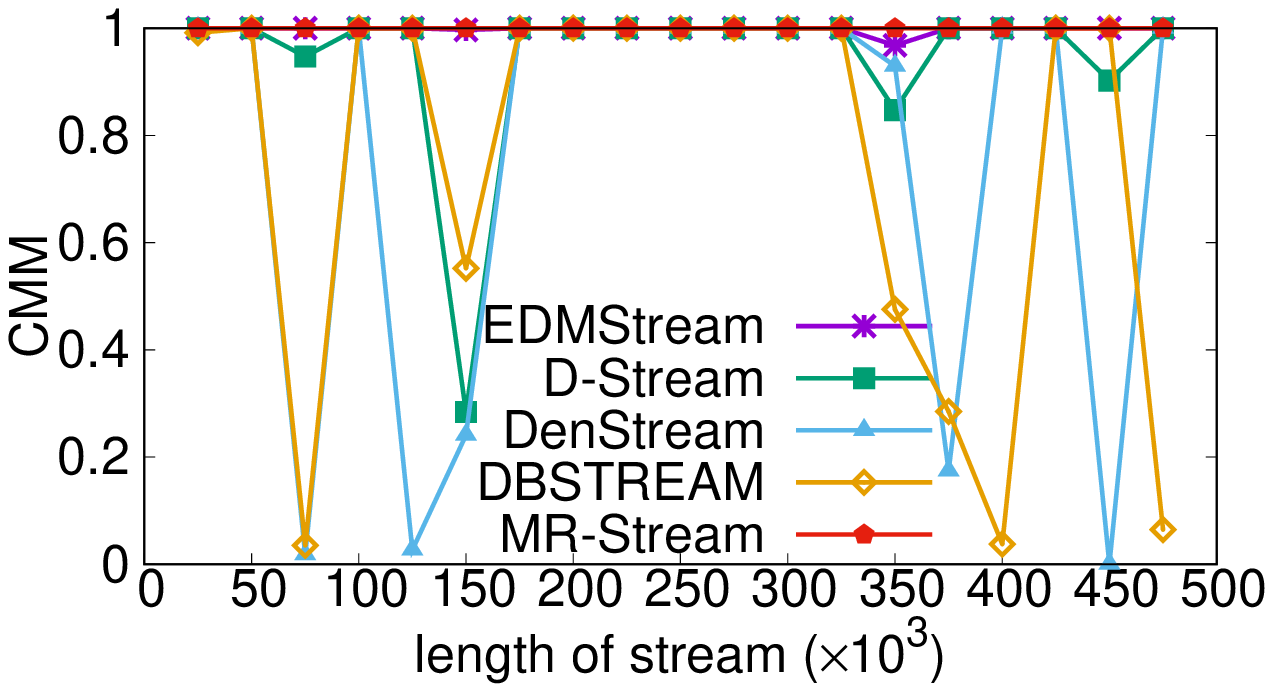}\label{fig:kddcup99result}}\ \
    \hspace{2pt}
    \subfloat[CoverType]{\includegraphics[width = 2in]{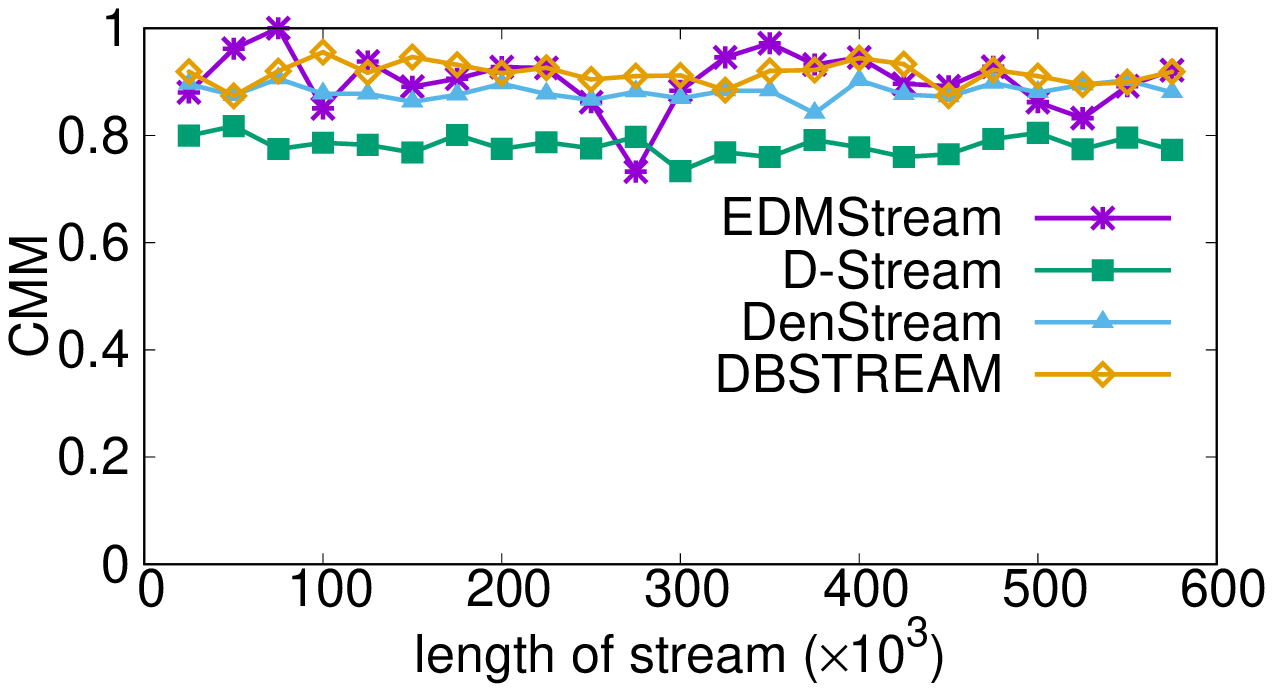}\label{fig:covtyperesult}}\ \
    \hspace{2pt}
    \subfloat[PAMAP2]{\includegraphics[width = 2in]{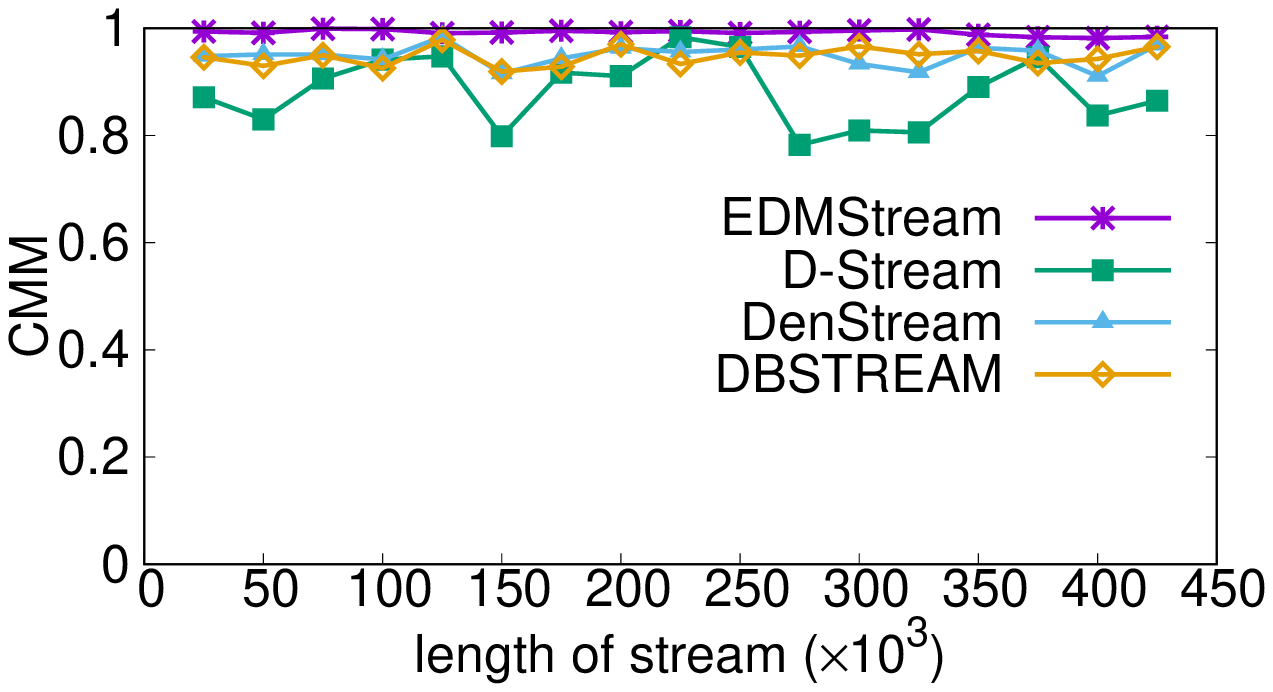}\label{fig:pparesult}}\ \
    \vspace{-0.12in}
    \caption{The comparison of clustering result quality}\label{fig:resultquality}
    \vspace{-0.22in}
\end{figure*}
Since we focus on cluster quality in this experiment, we reduce the stream points rate as much as possible to let competitors run normally. We launch these algorithms on three real datasets and compare their CMM \cite{Kremer2011An,Silva2014Data} metrics. Fig. \ref{fig:resultquality} shows their cluster CMM values over time. Our \textit{EDMStream} has comparable cluster quality with other algorithms. It is notable that MR-Stream treat each point as a cluster on the CoverType and PAMAP2 datasets, we do not consider it. Generally speaking, \textit{EDMStream}, DenStream, and DBSTREAM outperform D-Stream and MR-Stream in terms of CMM metric.
\vspace{-0.05in}
\subsubsection{Varying Stream Rate}\label{sec:cmmrate}

\begin{figure}[tp]
\vspace{-0.1in}
  \centering
   \includegraphics[width=2.2in, height=1.5in]{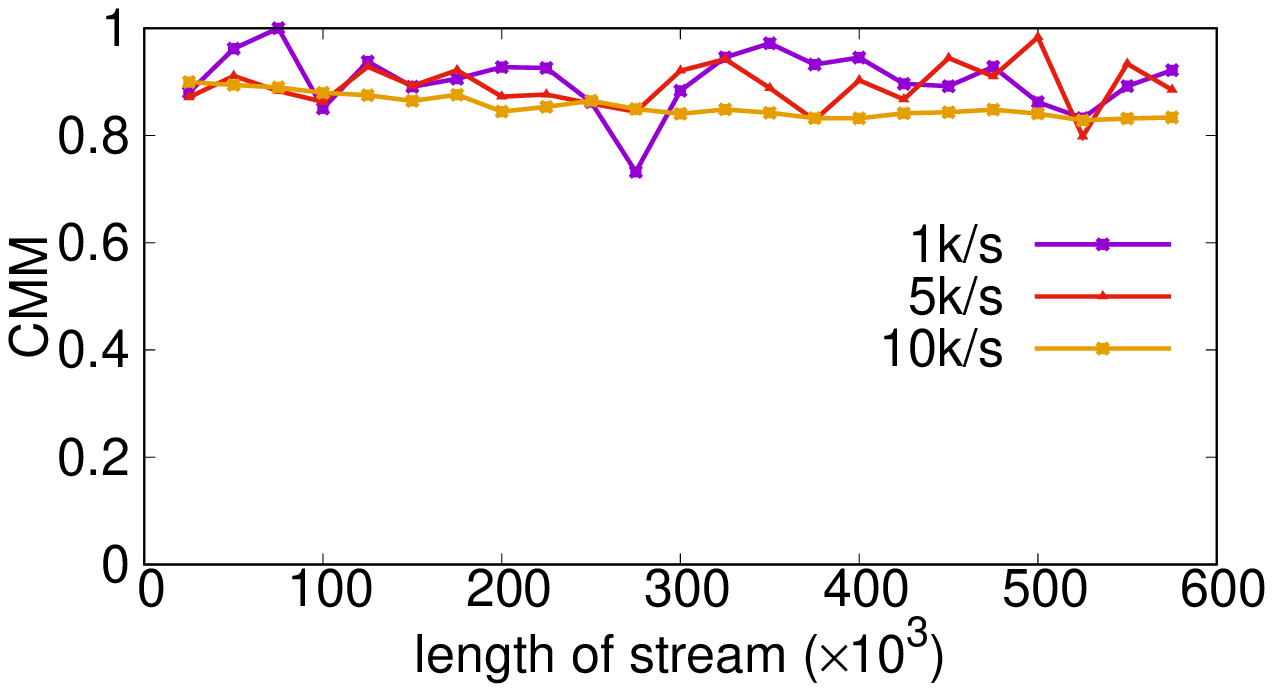}
   \vspace{-0.3in}
  \caption{Cluster quality varying stream rate}\label{fig:cmmrate}
  \vspace{-0.2in}
\end{figure}
We evaluate the cluster quality of \textit{EDMStream} by varying the stream rates (1K/s, 5K/s, 10K/s) on the CoverType dataset. Because the competitors fail with higher stream rate, we just compare CMM metrics of \textit{EDMStream} with different stream rates. The results are shown in Fig. \ref{fig:cmmrate}. We can see that the cluster quality are stable with various stream rates. \textit{EDMStream} is stable even in higher stream rate.




\vspace{-0.05in}
\subsection{Adaptability: Dynamic $\tau$ vs. Static $\tau$}
\label{sec:dynamicandstatic}

\textit{EDMStream} has the ability of adjusting itself to be adapt to the data distribution evolution (see Sec. \ref{sec:adaptive} and Fig. \ref{6result}). Precisely, \textit{EDMStream} can dynamically adjust the setting of the key parameter $\tau$, which controls the cluster separation granularity. In order to show the effectiveness of our dynamical $\tau$ setting strategy, we compare our \textit{dynamic} method with the \textit{static} method, which set $\tau$ as a constant.


\begin{table}[t]
\vspace{-0.05in}
\centering
\caption{The number of clusters changes over time (SDS)}\label{tab:tauquality}
\vspace{-0.1in}
\begin{tabular}{|c|c|c|c|c|c|c|c|c|c|c|} \hline
$t$ (s) & 1 & 2 & 3 & 4 & 5 & 6 & 7 & 8 & 9 & 10\\ \hline\hline
dynamic $\tau$& 2 & 2 & 2 & \color{blue}{2} & \color{blue}{2} & \color{blue}{2} & 1 & 1 & 1 & 1\\ \hline
static $\tau$  & 2 & 2 & 2 & \color{red}{1} & \color{red}{1} & \color{red}{1} & 1 & 1 & 1 & 1\\ \hline
\end{tabular}
\vspace{-0.1in}
\end{table}

\begin{figure}[t]
 \vspace{-0.1in}
    \captionsetup[subfigure]{farskip = -0.07in}
   \centering
   \subfloat[init $\tau$]{\includegraphics[width = 1.5in]{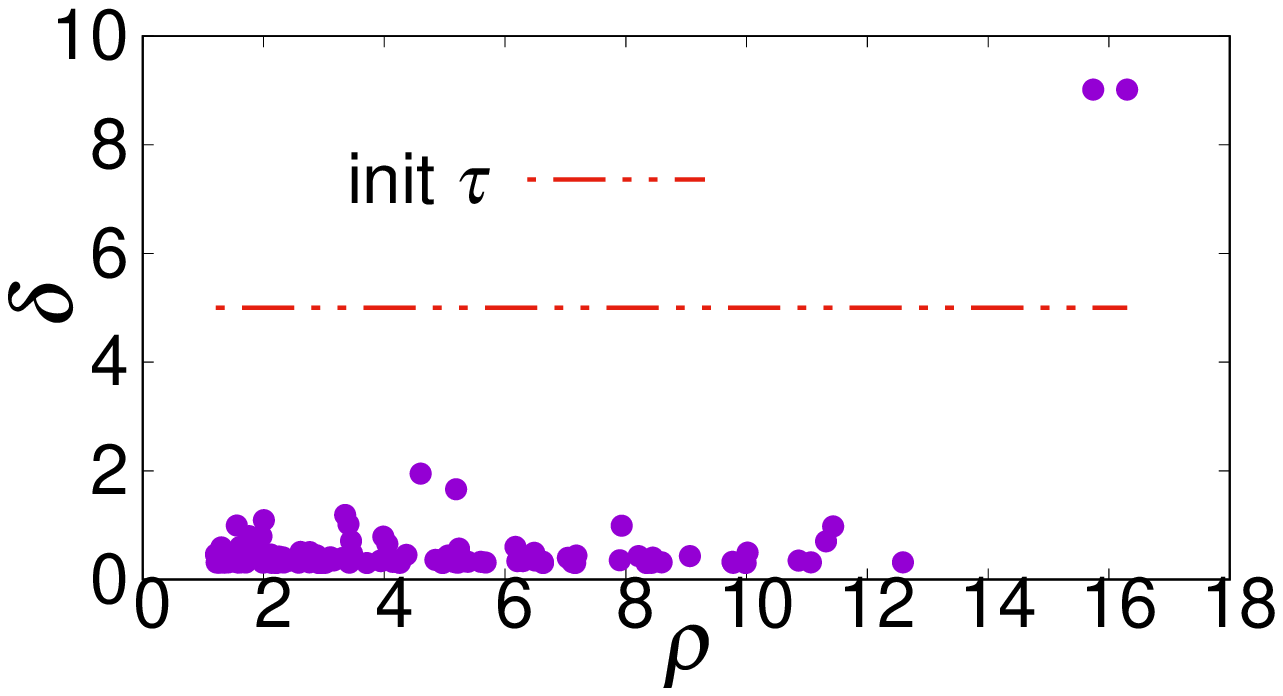}\label{fig:initGD}}\ \
   \vspace{-0.2in}
   \subfloat[$t=4s$]{\includegraphics[width = 1.5in]{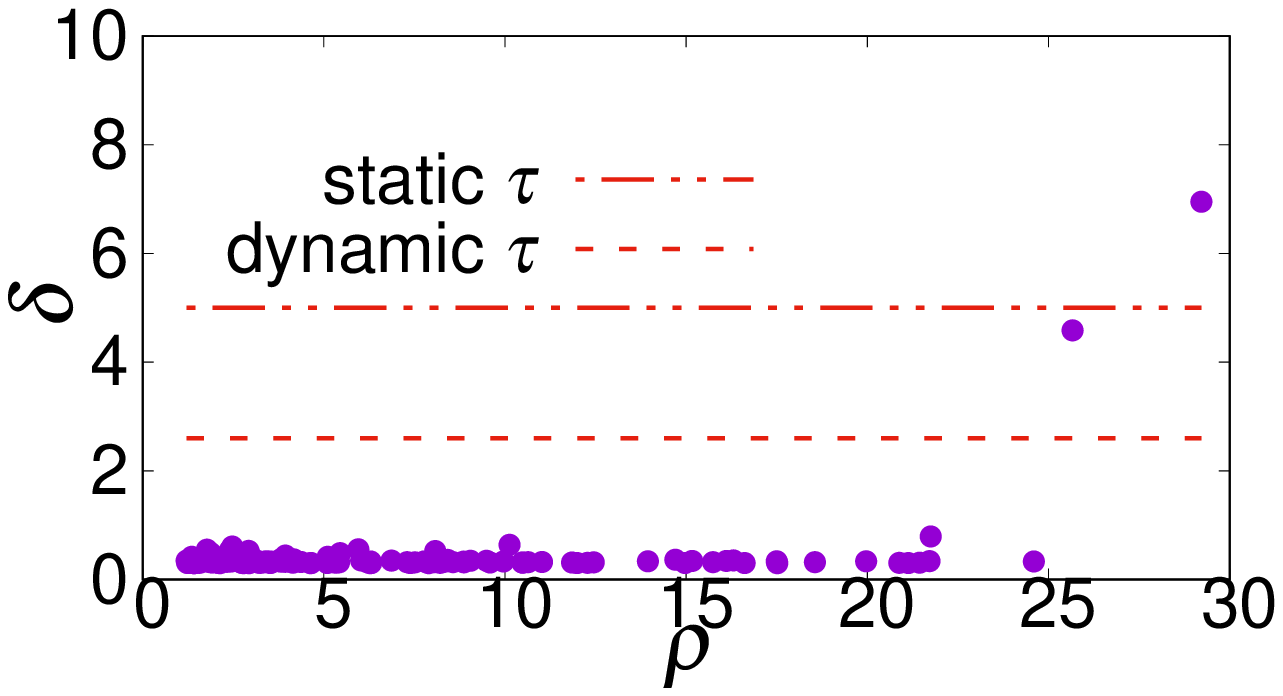}\label{fig:4DG}}\\
   \subfloat[$t=5s$]{\includegraphics[width = 1.5in]{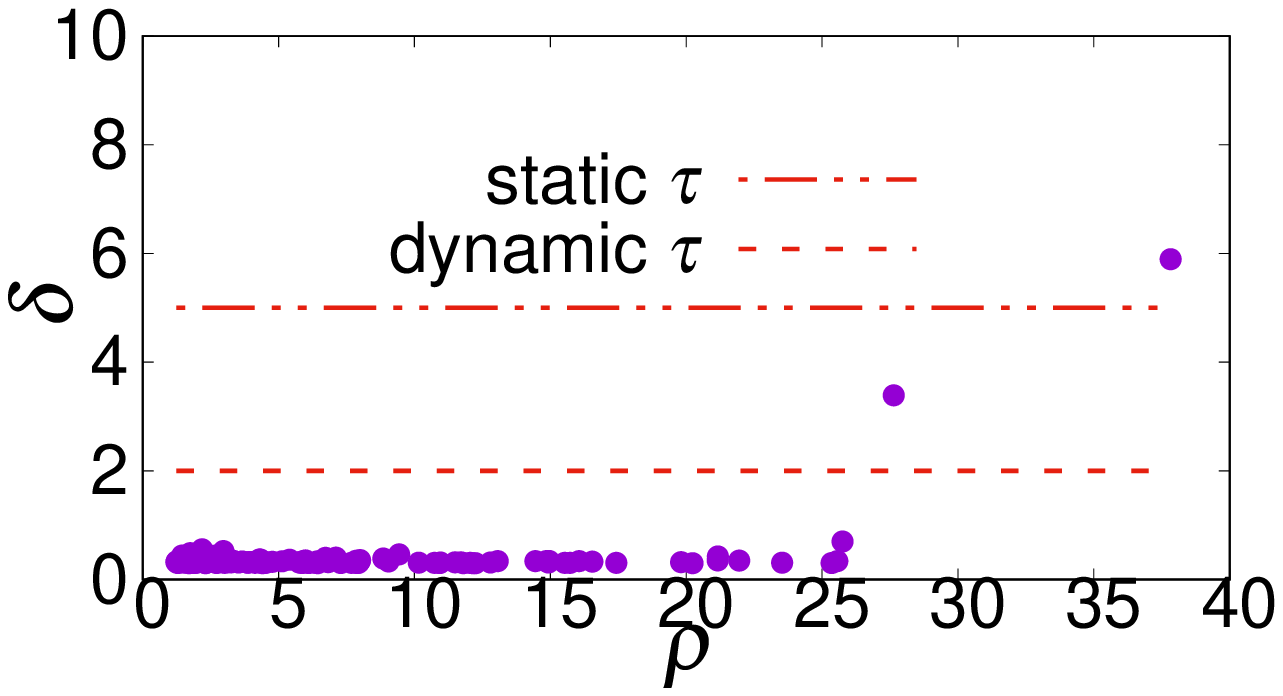}\label{fig:5DG}}\ \
   \subfloat[$t=6s$]{\includegraphics[width = 1.5in]{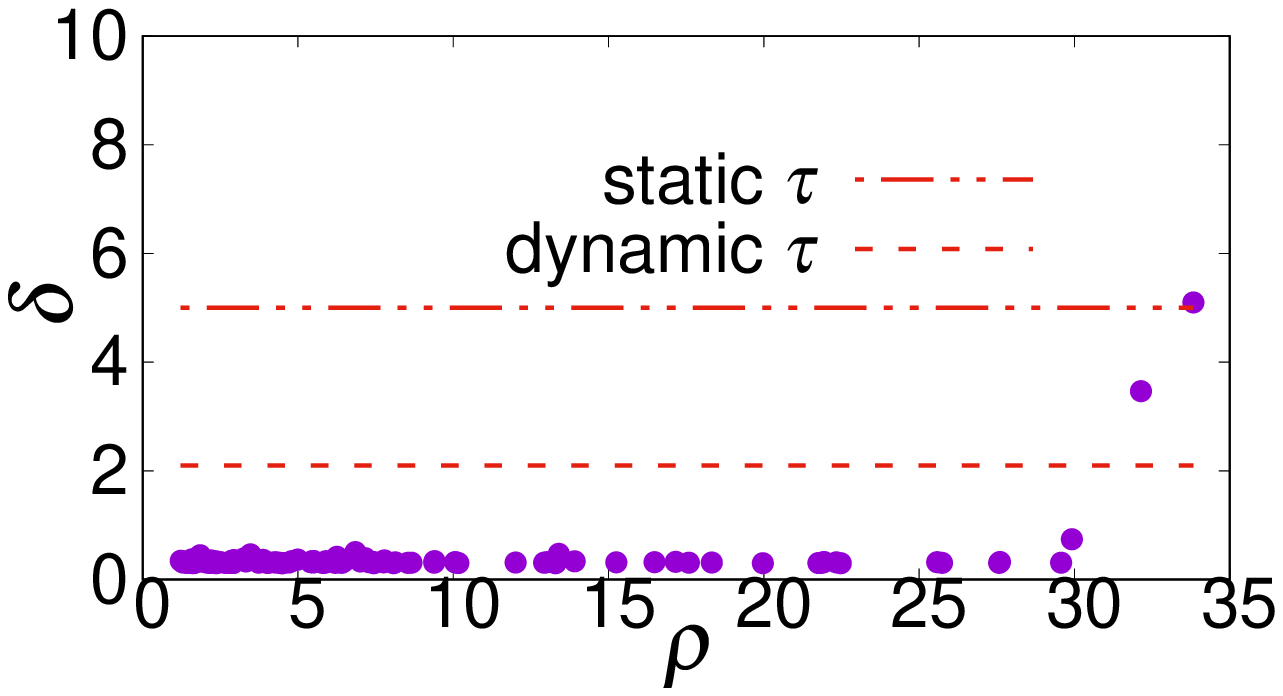}\label{fig:6DG}}\\
   \vspace{-0.15in}
    \caption{Decision graphs at different time points (SDS)}\label{fig:decgraphs}
    \vspace{-0.1in}
\end{figure}

We use the synthetic SDS to evaluate the \textit{dynamic} method and the \textit{static} method. Table \ref{tab:tauquality} shows the number of clusters in the first 10 seconds, where the 1st second result is the init result with user participation. The result of other time points is generated by algorithm automatically. We can see that the result is different at 4-6s. In order to identify the correct cluster result, we draw the decision graphs at init time and at these three time points as shown in Fig. \ref{fig:decgraphs}. In the init step, we choose $\tau=5$ since it can distinguish the density peaks from the other objects. Then the \textit{dynamic} method adjust $\tau$ value over time, while the \textit{static} method keeps using the same $\tau=5$.

As shown in Fig. \ref{fig:decgraphs}, at 4s-6s, two density peaks are higher than the dynamic $\tau$ line, but only one is higher than the static $\tau$ line. That is, two clusters exist at 4s-6s by dynamically adjusting $\tau$ while only one exists by using fixed $\tau$. Let us look at the original data distribution at 4s in Fig. \ref{t2}. It is obvious that the result of two clusters makes more sense. By using the static $\tau$, it will fail to identify the other density peak and as a result obtain a wrong clustering result. Therefore, our stream cluster algorithm has the self-adjustment ability, which is important to a stream clustering algorithm if running for a long period.
\vspace{-0.05in}
\subsection{Size of Outlier Reservoir}

\begin{figure}[t]
\vspace{0.05in}
    \centering
    \subfloat[CoverType]{\includegraphics[width = 1.5in]{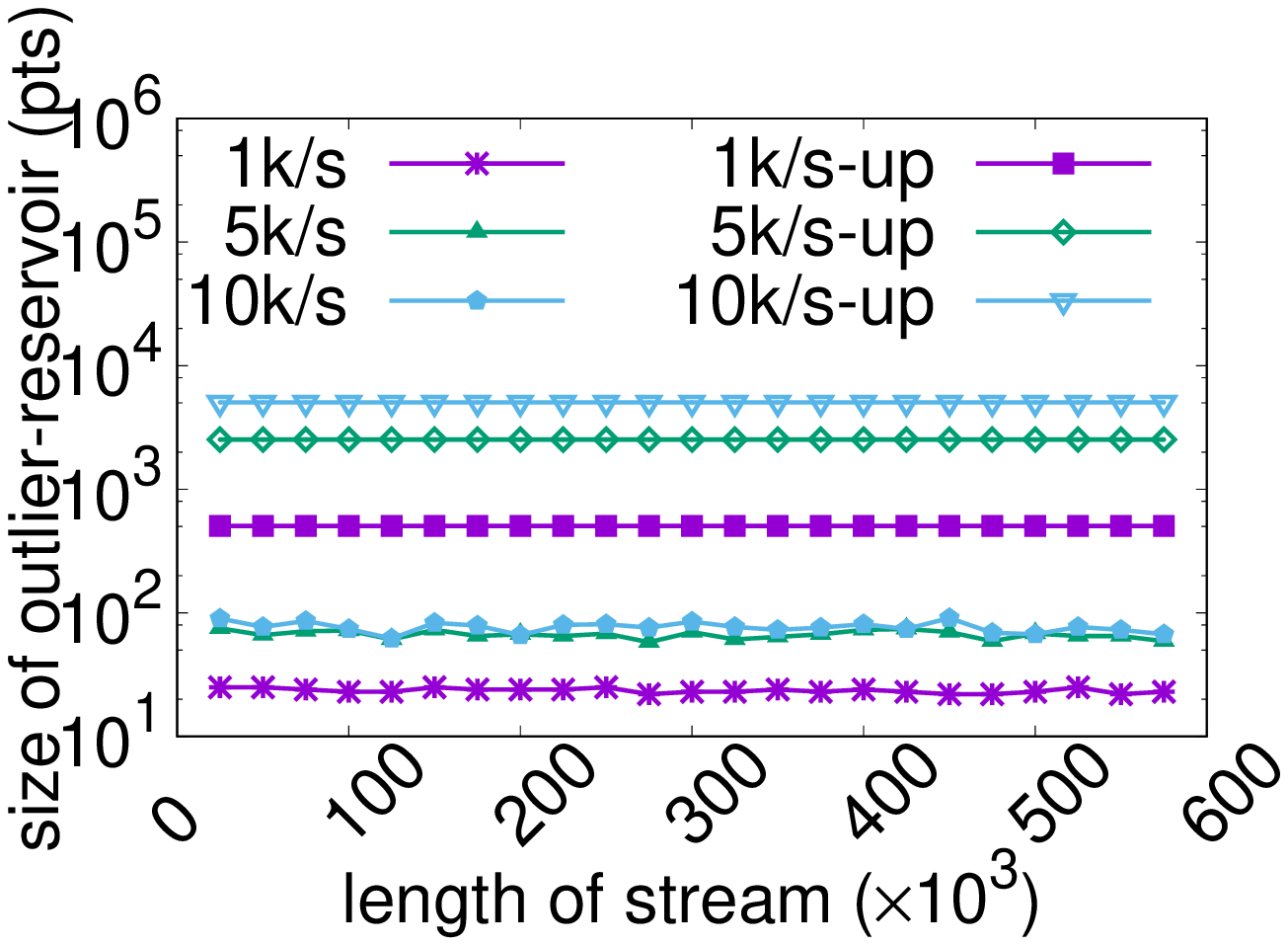}\label{fig:covtypeoutliers}}\ \
    \subfloat[PAMAP2]{\includegraphics[width = 1.5in]{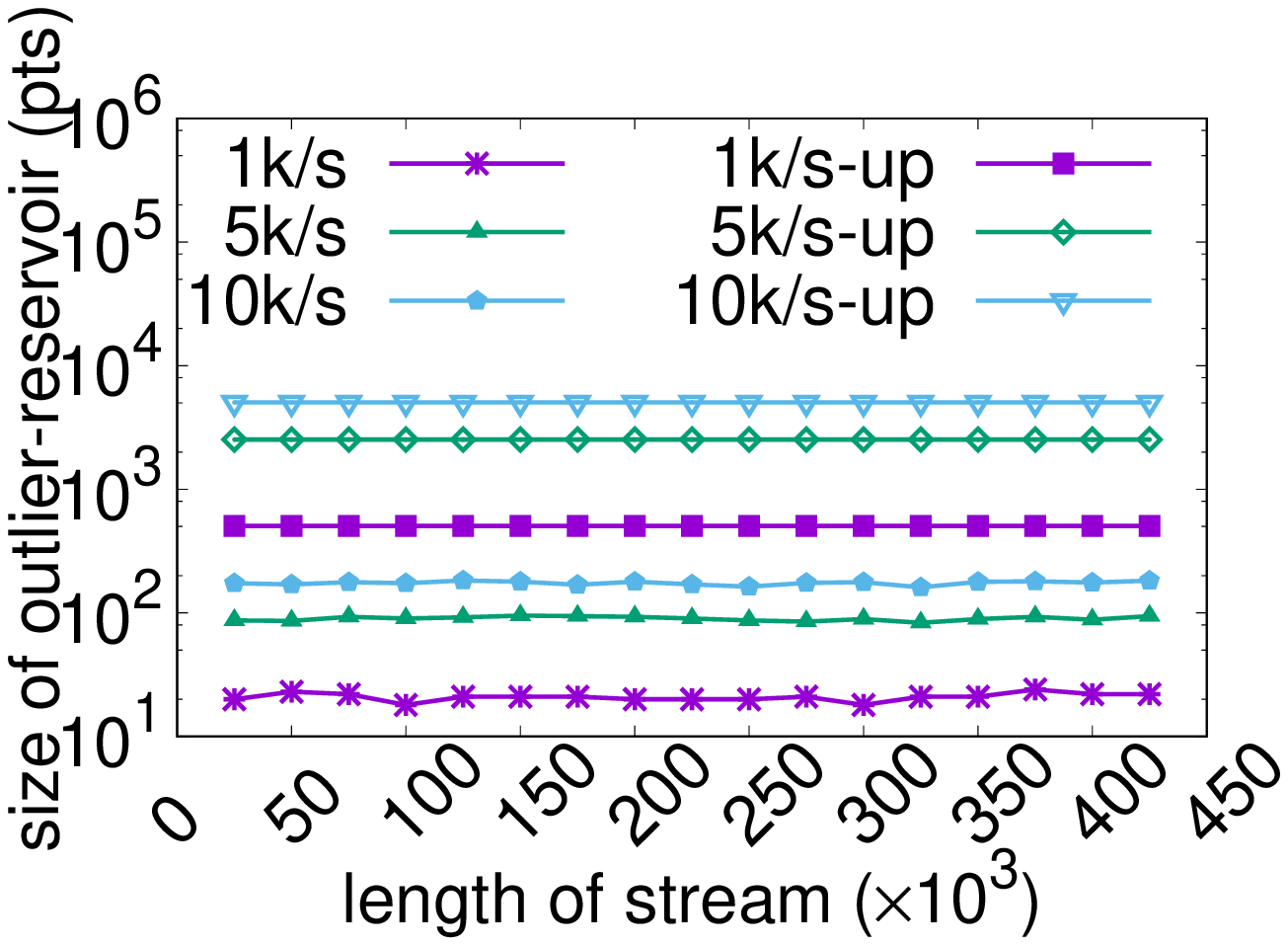}\label{fig:ppaoutliers}}
    \vspace{-0.1in}
    \caption{Size of outlier reservoir}\label{fig:outliers}
    \vspace{-0.1in}
\end{figure}

Tackling outliers is important for stream clustering, especially when massive outliers are mixed in stream. \emph{EDMStream} buffers the potential outliers (inactive cluster-cells) in the outlier reservoir. As discussed in Sec. \ref{subsec:remove}, given an average point arrival rate, the number of outliers has a theoretical upper bound in order to limit the size of outlier reservoir. In this experiment, we vary the point arrival rate (1K/s, 5K/s, 10K/s) and measure the outlier reservoir's size every X seconds. Fig. \ref{fig:outliers} shows the theoretical upper bounds of size (1K/s-up, 5K/s-up, 10K/s-up) and the measured sizes (1K/s, 5K/s, 10K/s). We see that the actual reservoir size is far less than the upper bound. The size of outlier reservoir can be predicted given the average point arrival rate. In order to reduce the size of outlier reservoir, users can accordingly adjust the decay model parameters to prolong point's freshness.

\vspace{-0.05in}
\subsection{Effect of Cluster-Cell Radius \large$r$}\label{subsec:radius}

The setting of cluster-cell radius $r$ has the effect on the cluster quality as well as processing speed. We test the clustering quality by varying $r$. As suggested in the original Density Peaks Clustering paper \cite{rodriguez2014clustering}, $r$ is chosen from 0.5\% to 2\% of the distance of all pairs of objects in the ascending order. Fig. \ref{fig:varyr} shows the cluster quality and response time when varying $r$. Smaller $r$ means more fine-grain cluster-cells. As result, it results in higher quality clusters but more computation overhead (i.e., longer response time). On the contrary, larger $r$ means less coarse-grain clusters so as to obtain lower quality clusters but return result faster.



\begin{figure}[t]
    \centering
    \subfloat[Cluster quality]{\includegraphics[width = 1.6in]{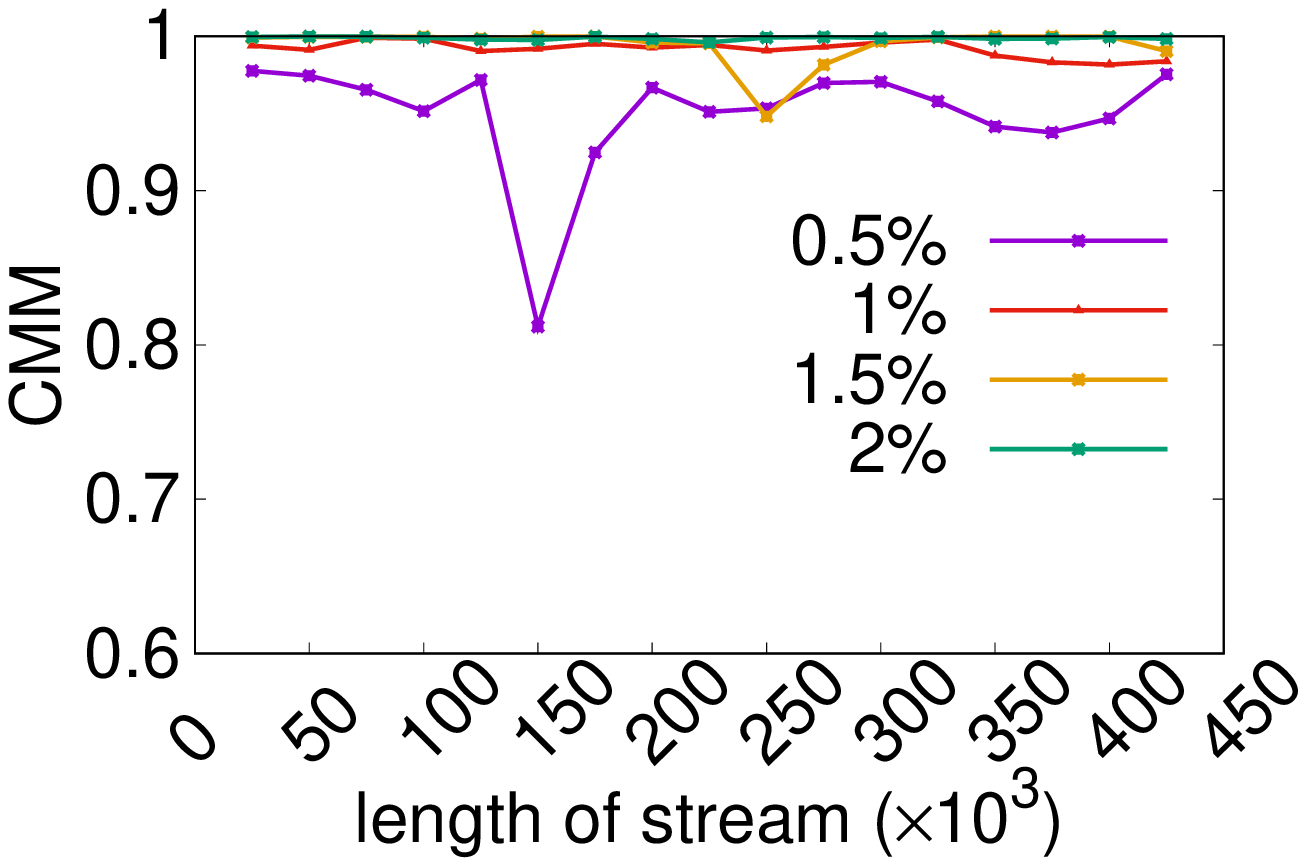}\label{fig:cmmvarRadius}}\ \
    \subfloat[Response time]{\includegraphics[width = 1.6in]{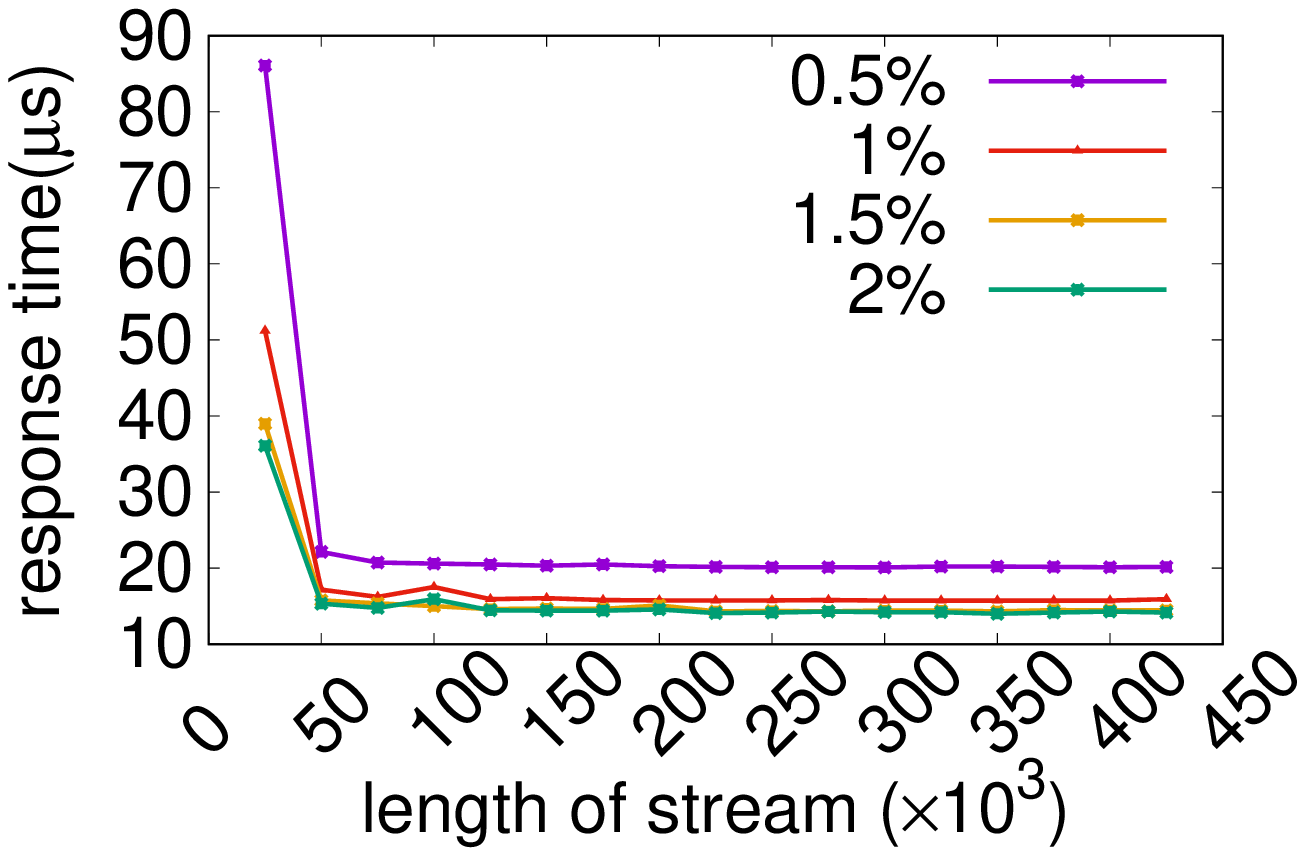}\label{fig:runtimevarRadius}}
    \vspace{-0.05in}
    \caption{Effect when varying $r$ (PAMAP2)}\label{fig:varyr}
    \vspace{-0.1in}
\end{figure}

\vspace{-0.05in}
\section{Related Work}
\label{sec:relate}

Clustering is one of the most important topics in data mining and has been very extensively studied. In recent years, real-time analysis and mining of stream data have attracted much attention from the research community. Many stream data clustering algorithms have also been proposed.

\vspace{-0.03in}
\Paragraph{Offline Clustering vs. Online Clustering} A class of stream clustering algorithms use a two-step approach, such as CluStream \cite{aggarwal2003framework}, DenStream \cite{Cao2006Density} and D-Stream \cite{Chen2007Density}, MR-Stream \cite{Wan2009Density} and DBSTREAM \cite{hahsler2016clustering}. The procedure of clustering contains an online data abstraction step and an offline clustering step. In the data abstraction step, the data stream is summarized using a specific data structure (e.g., micro-cluster and grid). The offline clustering step is triggered upon request, and a classical clustering method (e.g., k-means \cite{macqueen1967some} or DBSCAN \cite{ester1996density}) are used on these summarized data to obtain the cluster result. Our \textit{EDMStream} relies on online cluster result update, which can return the cluster result in real time.

\Paragraph{DBSCAN-based Clustering vs. DP-based Clustering}
A number of stream clustering algorithms are based on DBSCAN \cite{ester1996density} algorithm, e.g., DenStream \cite{Cao2006Density}, D-Stream \cite{Chen2007Density}, MR-Stream" \cite{Wan2009Density} and DBSTREAM \cite{hahsler2016clustering}. DBSCAN uses the density-connected information to build a density-connected graph. The clustering is to find all the \emph{maximal density-connected components} from the graph. DP clustering relies on the dependency relationship to build a DP-tree. The clustering is to find all the \emph{maximal strongly dependent subtrees} from the DP-Tree. As the tree (in DP) or graph (in DBSCAN) is continuously updating, we have to query the specific subtrees or subgraphs quickly according to the update. The \emph{hierarchical} tree structure is naturally more suitable for such a query than the \emph{flat} graph structure. This is because that it is much easier to identify the affected parts in tree than in graph. We only need to handle the affected successors of the update node in tree. However in graph structure, we have to re-evaluate the whole graph since there is no obvious affected part. This is why we rely on DP clustering and density mountain.
\vspace{-0.03in}
\vspace{-0.03in}
\Paragraph{CF-tree vs. DP-Tree} One of the earliest stream clustering algorithm is BIRCH \cite{zhang1996birch}. BIRCH uses \emph{cluster feature vector} (CF) to build CF-tree which abstract the agglomerative hierarchical clustering results. Though both \emph{EDMStream} and BIRCH use tree like structure to abstract clusters, i.e., DP-Tree and CF-Tree, these two trees are fundamentally different. DP-Tree is used to abstract the density mountains, where the nodes represents stream points or cluster-cells, and the links depict the dependency relationship between points. CF-Tree is used to manage the hierarchical clustering results, where the nodes represent certain grain level clusters, and the links depict the relationship between coarse-grain clusters and fine-grain clusters.

\vspace{-0.03in}
\Paragraph{Dynamic Clustering vs. Stream Clustering}
Recently, Gan and Tao \cite{gan2017dynamic} propose a new DBSCAN-based algorithm for dynamic clustering, which can return the updated clustering result very quickly (say in microseconds). The proposed dynamic clustering algorithm efficiently maintains data clusters along with point insertion/deletion occurring in the underlying dataset. While in stream clustering, we assume that data points are inserted continuously with timestamp information. The timestamp information is a key feature for streaming objects, which defines a freshness level. Data can also be decayed to meaningless levels. The time information reflects the emerging trends and is important for tracking cluster evolution. Although both dynamic clustering and stream clustering aim to return the updated clustering result in real time, stream clustering distinguishes fresh data from stale data and tends to give the fresh data more weight in clustering.


\vspace{-0.03in}
\Paragraph{Cluster Evolution Tracking}
All the above stream clustering algorithms fail to capture the cluster evolution activities. A few algorithms are proposed to monitor cluster evolution, e.g., MONIC \cite{Spiliopoulou2006MONIC} and MEC \cite{Oliveira2010MEC}. They trace the evolution of clusters by identifying the overlapping degree between two clusters. In MEC \cite{Oliveira2010MEC}, evolution tracking mechanism relies on the identification of the temporal relationship among them. In \textit{EDMStream}, we take full advantage of DP-Tree to track the evolutions of clusters. We can track the evolution according to the updates of DP-Tree structure.


\section{conclusion}\label{sec:conclu}
In this paper, we propose the \textit{EDMStream} algorithm, an effective and efficient method for clustering stream data. This method can not only track the evolution of stream data by monitoring the density mountain but also response cluster update in real time. By using the DP-Tree structure and a number of filtering schemes, the cluster result updates can be completed very quickly. \textit{EDMStream} also has the ability of adjusting itself to adapt to data evolution. This feature distinguishes it from most other stream clustering algorithms. Our experimental results demonstrate the effectiveness, efficiency, and adaptability of our algorithm. The quality of clustering is shown to be superior to the state-of-the-art stream clustering algorithms. \textit{EDMStream} can respond to a cluster update in 7-23 $\mu$s on average by using commodity PC and successfully adjust itself as data evolve.


\bibliographystyle{abbrv}
\bibliography{sigproc}

\end{document}